\documentclass[11pt]{amsart}
\usepackage{geometry}                
\usepackage{graphicx}
\usepackage{amssymb}
\usepackage{algorithm}
\usepackage{algorithmic}
\usepackage{tikz}
\usepackage{subcaption}

\usepackage[bookmarks=true,hypertexnames=false,pagebackref]{hyperref}
\hypersetup{colorlinks=true, citecolor=blue, linkcolor=red, urlcolor=blue}

\newcommand{\bfX}{\mathbf{X}}
\let\phi=\varphi
\newcommand{\calD}{\mathcal{D}}
\newcommand{\calI}{\mathcal{I}}
\newcommand{\calC}{\mathcal{C}}
\newcommand{\calT}{\mathcal{T}}
\newcommand{\nset}{\mathbb{N}}
\newcommand{\Zset}{\mathbb{Z}}

\newcommand{\numP}{\#{\bf P}}

\newcommand{\Ex}{\mathop{\mathbb{{}E}}\nolimits}
\newcommand{\scp}{{\sf Scp}}

\newcommand{\PRS}{\mathrm{PRS}}
\newcommand{\PRSforIS}{\mathrm{PRSforIS}}
\newcommand{\Sbar}{\overline S}

\newcommand{\eps}{\varepsilon}
\newcommand{\Garrow}{\vec G}

\renewcommand{\Pr}{\mathop{\mathrm{Pr}}\nolimits}

\tikzset{lab/.style={circle,thick,draw}}
\tikzstyle{vertex} = [draw,shape=circle,minimum size=4pt, inner sep=0pt, color=black]

\newtheorem{theorem}{Theorem}
\newtheorem{lemma}[theorem]{Lemma}

\newtheorem{observation}[theorem]{Observation}
\newtheorem{corollary}[theorem]{Corollary}
\newtheorem*{remark}{Remark}

\newtheorem{definition}[theorem]{Definition}

\usepackage{cleveref}

\crefname{theorem}{Theorem}{Theorems}
\crefname{observation}{Observation}{Observations}
\crefname{claim}{Claim}{Claims}
\crefname{condition}{Condition}{Conditions}
\crefname{algorithm}{Algorithm}{Algorithms}
\crefname{property}{Property}{Properties}
\crefname{example}{Example}{Examples}
\crefname{fact}{Fact}{Facts}
\crefname{lemma}{Lemma}{Lemmas}
\crefname{corollary}{Corollary}{Corollaries}
\crefname{definition}{Definition}{Definitions}
\crefname{remark}{Remark}{Remarks}
\crefname{proposition}{Proposition}{Propositions}
\crefname{equation}{equation}{equations}

\title{Fundamentals of Partial Rejection Sampling} 
\author{Mark Jerrum}
\address{School of Mathematical Sciences, Queen Mary, University of London, Mile End Road, London E1~4NS}
\thanks{This work was supported by EPSRC grant EP/S016694/1, ``Sampling in hereditary classes''.}

\begin{document}

\begin{abstract}
Partial Rejection Sampling is an algorithmic approach to obtaining a perfect sample from a specified distribution.  The objects to be sampled are assumed to be represented by a number of random variables.  In contrast to classical rejection sampling, in which all variables are resampled until a feasible solution is found, partial rejection sampling aims at greater efficiency by resampling only a subset of variables that `go wrong'.  Partial rejection sampling is closely related to Moser and Tardos' algorithmic version of the Lov\'asz Local Lemma, but with the additional requirement that a specified output distribution should be met.  This article provides a largely self-contained account of the basic form of the algorithm and its analysis.    
\end{abstract}

\maketitle

\section{The setting}
The aim of this expository article is to provide a uniform treatment of a particular approach to sampling combinatorial structures.  The method is a development of classical rejection sampling.  Suppose $\Phi(\bfX)$ is a predicate (Boolean function) depending on random variables $\bfX=(X_1,\ldots,X_n)$ coming from a product distribution. We would like to obtain a sample from the conditional distribution of~$\bfX$ given that $\Phi(\bfX)$ holds.  Classical rejection sampling repeatedly generates realisations of~$\bfX$ from the product distribution until one that satisfies $\Phi$ is found, and then outputs that.  In many situations this approach is very inefficient, as satisfying assignments to $\Phi$ may occur with exponentially small probability (in~$n$).  The idea in Partial Rejection Sampling (PRS) is to identify small subsets of the variables that violate $\Phi$ (in some sense), and resample just those variables.  Clearly, the choice of which variables to resample has to be done with care if the output distribution is to remain the correct one.  

The phrase `partial rejection sampling' appears to have been coined by Cohn, Pemantle and Propp~\cite{CPP02} to describe their approach to sampling sink-free orientations.  Noting the similarity to Wilson's approach to sampling spanning trees~\cite{PW98}, they wondered whether one could develop a general theory.  PRS as a general algorithmic technique was explored by Guo, Jerrum and Liu~\cite{GJL19}, and is our topic here.

We focus on examples of PRS inspired by algorithmic proofs of the Lov\'asz Local Lemma (LLL).  There is a substantial literature on this topic to assist us, but it is concerned only with problem of constructing \emph{some} satisfying assignment to~$\Phi$.  Partial rejection sampling adds the novel requirement that the output should be uniform (or, more generally, from the desired distribution) on satisfying assignments.  This additional requirement adds a new challenge.   

In the study of the LLL, the class of `extremal' instances receives particular attention.  The extremal instances are particularly suited to PRS and we treat them first.  It is unlikely that anything in this section of the article is conceptually new.  However, the treatment of PRS in the extremal regime involves some particularly beautiful combinatorial ideas, and it is seems a good time to bring this material together in one place, with a consistent approach and notation.    

After that, we investigate to what extent the conditions defining extremal instances can be relaxed.  The viewpoint taken in this section is novel to a certain extent.  Care has been taken to set out the conditions under which PRS continues to function in the non-extremal setting, in the hope that it will help in discovering new applications.  More than usual attention is paid to the flexibility in the order in which variables can be resampled.

The scope of this article is limited to versions of partial rejection sampling that stay close to the spirit of the algorithmic LLL pioneered by Moser and Tardos~\cite{MT10}.  In particular, the number of random variables is finite and all constraints on them are `hard'.  Informally, we restrict attention to a `combinatorial' setting, which excludes important applications to spin systems in statistical physics.  We finish with a few pointers to work that goes beyond the framework presented here.  

It should be noted that PRS is not the only approach to perfect sampling.  Perhaps the best known and most extensively studied is `Coupling From The Past' (CFTP), which was pioneered by Propp and Wilson~\cite{PW98}. Other approaches include Fill and Huber's `Randomness recycler'~\cite{FillHuber} and Anand and Jerrum's `Lazy  depth-first sampler'~\cite{AnandJerrum}.

\section{Partial rejection sampling in the extremal setting}\label{sec:extremal}

Suppose $\Phi=\phi_1\wedge\phi_2\wedge\cdots\wedge\phi_m$ is formula on variables $\bfX=(X_1,X_2,\ldots,X_n)$.
Each clause~$\phi_k$ depends on a certain tuple $(X_{i_1},\ldots,X_{i_{a_k}})$ of variables, where $a_k$ is the \emph{arity} of~$\phi_k$.  We refer to the tuple $\scp(\phi_k)=(i_1,\ldots,i_{a_k})$ of indices as the \emph{scope} of the clause~$\phi_k$.  By assuming $i_1<i_2<\cdots<i_{a_k}$ we can consider the scope $\scp(\phi_k)$ to be a $a_k$-tuple or a set of cardinality~$a_k$, according to context.  For a set $S\subseteq\{1,\ldots,n\}$ of indices we write $X_S=\{X_i:i\in S\}$.  Then to emphasise the dependence on the variables we can write 
$$\Phi(\bfX)=\phi_1(X_{\scp(\phi_1)})\wedge\cdots\wedge \phi_m(X_{\scp(\phi_m)}).$$
\begin{definition}\label{def:extremal}
We say that the formula $\Phi=\phi_1\wedge\cdots\wedge\phi_m$ is \emph{extremal} if, for all $1\leq k<\ell\leq m$ satisfying $\scp(\phi_k)\cap \scp(\phi_\ell)\not=\emptyset$, it is the case that $\phi_k(X_{\scp(\phi_k)})\vee\phi_\ell(X_{\scp(\phi_\ell)})$ is a tautology.  In other words, any two clauses that are probabilistically dependent cannot both be false.  
\end{definition}
In this section we consider only extremal instances, as these can be dealt with using the basic form of partial rejection sampling.

Now suppose that variable $X_i$ takes values from a set~$D_i$.  Equip $D_i$ with a probability distribution and call the probability space $\calD_i$.  We are interested in sampling a realisation of the random variable $\bfX$ from the product distribution $\calD=\calD_1\times\calD_2\times\cdots\times\calD_n$ conditioned on $\Phi(\bfX)$ holding.  Denote this desired distribution by $\calD_\Phi$.  Partial Rejection Sampling (PRS) is a simple algorithm for accomplishing this task in the context of extremal instances.  It runs as follows.

\begin{algorithm}
\caption{Partial Rejection Sampling}\label{alg:PRS}
\begin{algorithmic}
\STATE $\PRS(\Phi,\calD)$ 
\COMMENT {$\Phi$ is a formula on variable set $\bfX$}
\STATE {Sample $\bfX$ from the product distribution $\calD=\calD_1\times\cdots\times\calD_n$}
\WHILE {$\neg\Phi(\bfX)$}
\STATE {Choose any clause $\phi_k$ with $\neg\phi_k(X_{\scp(\phi_k)})$}  
\STATE {Resample all variables in $\scp(\phi_k)$}
\ENDWHILE
\end{algorithmic}
\end{algorithm}

In the resampling step, the product distribution $\prod_{i\in\scp(\phi_k)}\calD_i$ is naturally being used.  

The algorithm PRS was first introduced by Moser and Tardos~\cite{MT10} in the context of an algorithmic proof of the Lov\'asz Local Lemma.  Its application to sampling from naturally specified distributions was studied by Guo, Jerrum and Liu~\cite{GJL19}, who analysed its correctness and efficiency.  Although their investigation seems to be the first attempt to treat PRS as a general technique, several specific examples had previously appeared in the literature as we noted above.

\begin{remark}
In the combinatorial community, the Moser and Tardos algorithm would be viewed as gradually eliminating the set of `bad events' until none are left.  In the area of constraint satisfaction, the goal is to simultaneously satisfy a collection of constraints.  It is important to keep in mind that, of these diametrically opposing conventions, we use the latter here.
\end{remark}

In classical rejection sampling we would resample the whole of $\bfX$ on each iteration.  In contrast, PRS resamples only a subset of offending variables.  We cannot expect the correctness of the algorithm to survive such extreme corner cutting.  Indeed, for general formulas~$\Phi$, the call $\PRS(\Phi,\calD)$ does not produce a sample from the distribution~$\calD_\Phi$.  Surprisingly, PRS \emph{does} achieve the desired distribution for extremal instances.

\begin{theorem}\label{thm:PRScorrect}
Suppose $\Phi$ is a satisfiable extremal instance.  Then $\PRS(\Phi,\calD)$ terminates with probability~1.  On termination, $\bfX$ is a realisation of a random variable from the distribution~$\calD_\Phi$.
\end{theorem}

To analyse the algorithm, we need to introduce time explicitly.  A \emph{resampling table} is a semi-infinite matrix $(\omega_{i,j}:1\leq i\leq n\text{ and } j\in\nset)$.  Each entry $\omega_{i,j}$ in the table is an independent sample from the distribution~$\calD_i$.  Fixing $i$, the sequence $\omega_{i,0}, \omega_{i,1},\omega_{i,2},\ldots$ specifies the sequence of values taken by the random variable~$X_i$ during the execution of the algorithm.  Introducing a superscript to indicate the time~$t$ (each iteration of the loop takes one time unit), we write $X_i^t=\omega_{i,j(i,t)}$.  If $X_i$ is resampled during iteration~$t$ then $j(i,t)=j(i,t-1)+1$, otherwise $j(i,t)=j(i,t-1)$.  Initially, $j(i,0)=0$ for all $1\leq i\leq n$.  (By convention, we start at time~0, and iteration $t$ occupies the interval between time $t-1$ and time~$t$.)  At any time $t$, the \emph{frontier} of the resampling table is $F(t)=(j(1,t),j(2,t),\ldots, j(n,t))$.  See Figure~\ref{fig:resampling}.

\begin{figure}[t]
\renewcommand{\arraystretch}{1.3}
\centering
\begin{tabular}{|c c c c c|} 
$\vdots$&$\vdots$&$\vdots$&$\vdots$&$\vdots$\\
$\omega_{1,4}$&$\omega_{2,4}$&$\omega_{3,4}$&$\omega_{4,4}$&$\omega_{5,4}$\\
$\omega_{1,3}$&$\omega_{2,3}$&$\omega_{3,3}$&$\omega_{4,3}$&$\omega_{5,3}$\\
$\omega_{1,2}$&$\omega_{2,2}$&$\omega_{3,2}$&$\omega_{4,2}$&$\omega_{5,2}$\\
$\omega_{1,1}$&$\omega_{2,1}$&$\omega_{3,1}$&$\omega_{4,1}$&$\omega_{5,1}$\\
$\omega_{1,0}$&$\omega_{2,0}$&$\omega_{3,0}$&$\omega_{4,0}$&$\omega_{5,0}$\\
\hline
$X_1$&$X_2$&$X_3$&$X_4$&$X_5$\\
\hline
\end{tabular}
\caption{A resampling table}
\label{fig:resampling}
\end{figure}

As time progresses, we record the actions of the algorithm in the form of a partition of the portion of the resampling table that lies behind the frontier, namely $(\omega_{i,j}:1\leq i\leq n\text{ and }0\leq j<j(i,t))$.  The partition builds as the frontier advances.  In iteration $t$, the variables in the scope of some clause $\phi_k$ are resampled.  The locations $\{(i,j(i,t-1)):i\in \scp(\phi_k))\}$ that were on the frontier now lie behind it;  this set of locations now forms a new block of the partition.  We call these blocks the \emph{resampling blocks}. At time~$t$, the frontier together with the partition into resampling blocks forms a \emph{transcript} of the run of the algorithm up to time~$t$.  

By way of example, consider the formula 
\begin{equation}\label{eq:Phi}
\Phi(\bfX)=(X_1\vee X_2)\wedge(\neg X_1\vee X_3\vee\neg X_4)\wedge(\neg X_2\vee\neg X_3\vee X_5)\wedge(X_4\vee \neg X_5)
\end{equation}
on variables $\bfX=(X_1,X_2,X_3,X_4,X_5)$.  (The formula $\Phi$ encodes sink-free orientations of a certain 4-vertex graph, a point we shall return to later.)  Thus $\phi_2=\neg X_1\vee X_3\vee\neg X_4$ and $\scp(\phi_2)=\{1,3,4\}$, and similarly for the other clauses.  A particular realisation of the resampling table that leads to termination of the algorithm PRS, together with its associated transcript, are depicted in Figure~\ref{fig:transcript}.  In the pictorial representation of the transcript, the values in the resampling table are spread out along the columns so that each resampling block of the transcript occupies a single row.  The rectangle at the top denotes the \emph{final frontier}, i.e., the frontier at termination. Initially, $X_4=0$ and $X_5=1$, which violates clause~$\phi_5$.  Accordingly, variables $X_4$ and $X_5$ are resampled, and $\{(4,0),(5,0)\}$ becomes the first resampling block of the transcript.  The value of $X_4$ switches from 0 to~1, and this causes clause $\phi_2$ to be violated, since now $X_1=X_4=1$ and $X_3=0$.  So $X_1$, $X_3$ and $X_4$ are resampled and $\{(1,0),(3,0),(4,1)\}$ becomes the next resampling block of the transcript.  Eventually, $\bfX=(0,1,0,1,1)$ , which satisfies~$\Phi$, and the algorithm halts.  

Suppose we run the algorithm PRS twice, using different non-deterministic choices (of which clauses to resample), until termination.  A priori, it might be imagined that the two runs would in general have different transcripts, but this is not the case, as we shall see in Lemma~\ref{lem:confluence}.  Some intuition can be gained from Figure~\ref{fig:transcript}. At time~4, $\bfX=(0,0,0,0,1)$, and hence clauses $\phi_1$ and $\phi_4$ are both violated.  We can resample either $\{X_1,X_2\}$ first or $\{X_4,X_5\}$, but either way we end up with the \emph{same} transcript.  In this context, it is crucial that $\scp(\phi_1)\cap\scp(\phi_4)=\emptyset$, but in an extremal instance, this condition is guaranteed. 

\begin{figure}[t]

\begin{minipage}[c]{0.44\linewidth}
\vglue 0pt
\renewcommand{\arraystretch}{1.3}
\centering
\begin{tabular}{|c c c c c|} 
$\vdots$&&&$\vdots$&\\
0&$\vdots$&$\vdots$&1&$\vdots$\\
1&1&0&1&1\\
0&0&0&0&1\\
0&1&1&1&0\\
1&0&0&0&1\\
\hline
$X_1$&$X_2$&$X_3$&$X_4$&$X_5$\\
\hline
\end{tabular}
\end{minipage}%
\begin{minipage}[c]{0.55\linewidth}
\vglue 0pt
\centering
\begin{tikzpicture}[xscale=0.08, yscale=0.06, line width=3pt, inner sep=2pt]]
    \draw (40,0)  node[lab] (a4) {0}; \draw (50,0)  node[lab] (a5) {1};
    \draw (10,10) node[lab] (b1) {1}; \draw (30,10) node[lab] (b3) {0}; \draw (40,10) node[lab] (b4) {1};
    \draw (10,20) node[lab] (c1) {0}; \draw (20,20) node[lab] (c2) {0};
    \draw (20,30) node[lab] (d2) {1}; \draw (30,30) node[lab] (d3) {1}; \draw (50,30) node[lab] (d5) {0};
    \draw (10,40) node[lab] (e1) {0}; \draw (20,40) node[lab] (e2) {0}; 
    \draw (40,40) node[lab] (e4) {0}; \draw (50,40) node[lab] (e5) {1};
    \draw (10,50) node[lab] (f1) {1}; \draw (30,50) node[lab] (f3) {0}; \draw (40,50) node[lab] (f4) {1};
  
    \draw (-2,60) node[text width = 2cm] () {$t=7$};  
    \draw (-2,50) node[text width = 2cm] () {$t=6$};
    \draw (-2,40) node[text width = 2cm] () {$t=4,5$};
    \draw (-2,30) node[text width = 2cm] () {$t=3$};
    \draw (-2,20) node[text width = 2cm] () {$t=2$};
    \draw (-2,10) node[text width = 2cm] () {$t=1$};
    \draw (-2,0)  node[text width = 2cm] () {$t=0$};

    \draw (73,50) node[text width = 2cm] () {$\phi_2$};
    \draw (73,40) node[text width = 2cm] () {$\phi_1,\phi_4$};
    \draw (73,30) node[text width = 2cm] () {$\phi_3$};
    \draw (73,20) node[text width = 2cm] () {$\phi_1$};
    \draw (73,10) node[text width = 2cm] () {$\phi_2$};
    \draw (73,0)  node[text width = 2cm] () {$\phi_4$};

    \draw (11,-10) node () {$X_1$};
    \draw (21,-10) node () {$X_2$};
    \draw (31,-10) node () {$X_3$};
    \draw (41,-10) node () {$X_4$};
    \draw (51,-10) node () {$X_5$};
    
    \draw (10,60) node () {0};
    \draw (20,60) node () {1};
    \draw (30,60) node () {0};
    \draw (40,60) node () {1};
    \draw (50,60) node () {1};
    
    \draw[thick] (5,56) rectangle (55,64) {};
    
    \draw[gray] (a4) -- (a5);
    \draw[gray] (b1) -- (b3) -- (b4);
    \draw[gray] (c1) -- (c2);
    \draw[gray] (d2) -- (d3) -- (d5);
    \draw[gray] (e1) -- (e2);
    \draw[gray] (e4) -- (e5);
    \draw[gray] (f1) -- (f3) -- (f4);
\end{tikzpicture}
\end{minipage}
\caption{A realisation of a resampling table, and the corresponding transcript}
\label{fig:transcript}
\end{figure}

\begin{lemma}\label{lem:confluence}Let $\Phi$ be an extremal formula.  Fix a resampling table.  Suppose that for some sequence of non-deterministic choices, $\PRS(\Phi,\calD)$ terminates with a certain transcript.  Then for any other sequence of choices, the algorithm will terminate with the same transcript.  
\end{lemma}

To prove this lemma, we use a version of Newman's Lemma that is particularly convenient in this application.  An (abstract) rewriting system is simply a set~$\calT$ of `positions' together with a binary `rewriting' relation $\to$ on~$T$.  For positions $t,s\in\calT$, the relation $t\to s$ indicates that it is possible to go from $t$ to~$s$ in one move.  A position~$t$ from which no valid move $t\to s$ is possible is said to be \emph{terminal}.  A sequence of moves ending at a terminal state is said to be \emph{terminating}.  Following Eriksson~\cite{Eriksson96}, we say that the rewriting system $(\calT,\to)$ has the \emph{polygon property} if, given any position $t\in\calT$ and two moves $t\to s$ and $t\to s'$, either (a)~there are two sequences $s=s_0\to s_1\to\cdots\to s_\ell=t^*$ and $s'=s'_0\to s'_1\to\cdots\to s'_\ell=t^*$ of the same length~$\ell$ that end at the same position~$t^*$, or (b)~there are two infinite sequences of moves starting from $s$ and~$s'$.  A rewriting system is said to have the \emph{strong convergence property} if, for any starting position~$t$ from which there exists a sequence of moves terminating at some position~$t^*$, it is the case that every sequence of moves starting from~$t$ will lead to~$t^*$, and in the same number of moves.  Eriksson~\cite[Thm~2.1]{Eriksson96} showed the following.

\begin{lemma}  \label{lem:Newman}
A rewriting system has the strong convergence property iff it has the polygon property.
\end{lemma}

\begin{proof}[Proof of Lemma~\ref{lem:confluence}]
Fix a resampling table.  View the collection of all possible transcripts as an abstract rewriting system by introducing a binary relation $\to$ on transcripts.  The meaning of $t\to s$ is that $s$ can follow $t$ in one iteration of $\PRS$.  This rewriting system has the diamond property, namely if $t\to s$ and $t\to s'$ then there exists $t^*$ such that $s\to t^*$ and $s'\to t^*$.  (This diamond property is clearly stronger than the polygon property.)  For suppose $t\to s$ is a result of resampling the variables in scope $\scp(\phi_k)$, and $t\to s'$ the result of resampling $\scp(\phi_\ell)$.  Since the instance~$\Phi$ is extremal we know that $\scp(\phi_k)\cap \scp(\phi_\ell)=\emptyset$.  Thus, we can resample whichever scope was not resampled in the first step, to get to a common transcript~$t^*$.  The result now follows from Lemma~\ref{lem:Newman}.
\end{proof}

\begin{proof}[Proof of Theorem~\ref{thm:PRScorrect}]
Fix a particular satisfying assignment $\bfX=\mathbf{b}=(b_1,\ldots,b_n)$ to $\Phi$.  At any point in the execution of the algorithm, the following fortuitous sequence of events may occur over the next $n$ iterations:  each time a variable $X_i$ is resampled, it is assigned the value~$b_i$.  On each iteration, $\bfX$ approaches closer to $\mathbf b$ in Hamming distance.  Thus, the algorithm will terminate in the next $n$~iterations.  Since this fortuitous sequence of events occurs with probability bounded away from~0, the running time of the algorithm PRS is stochastically dominated by an exponential random variable with finite mean. So the algorithm terminates with probability~1.  

Fix a resampling table $T$, and run $\PRS$ on $T$ to obtain a transcript.  Since the algorithm has terminated, we know that the frontier contains a satisfying assignment.  Create a new resampling table~$T'$ by replacing the values in the frontier by some other satisfying assignment.   Now run the algorithm again on $T'$ with the same nondeterministic choices of scopes to resample.  Note that this is always possible:  whenever the algorithm running on~$T$ resamples $\scp(\phi_k)$ at time~$t$ it is because $\phi_k(\omega_{i,j(i,t)}:i\in\scp(\phi_k))$ is false.  None of the resampled variables are in the final frontier, since no variables beyond the final frontier are ever inspected.  So the clause $\phi_k$ is also false when the algorithm is run on table~$T'$, and it is valid step to resample $\scp(\phi_k)$.  Finally, on the same iteration that the algorithm terminates when run on table~$T$, it will also terminate on~$T'$.  The same transcript (i.e., frontier $F(t)$ together with the partition of the table behind the frontier) arises from running the algorithm on~$T'$ as the one that arose from the run on~$T$.  

By Lemma~\ref{lem:confluence} any sequence of non-deterministic choices made by the algorithm on table~$T'$ leads to the same transcript.  Summarising, the final transcript does not depend on the nondeterministic choices made by the algorithm, and is also unchanged if one satisfying assignment is substituted for another in the final frontier.  Thus, conditioned on the transcript, each satisfying assignment $\bfX=(b_1,\ldots,b_n)$ of $\Phi$ occurs with  probability proportional to $\calD_1(b_1)\calD_2(b_2)\cdots\calD_n(b_n)$.   So, at termination, $\bfX$ is distributed as $\calD_\Phi$. 
\end{proof}

\begin{figure}[t]
\centering
\begin{tikzpicture}[xscale=0.12, yscale=0.12, inner sep=1pt]

    \draw (0,10)  node[lab] (v1) {$1$}; 
    \draw (10,20) node[lab] (v2) {$2$};
    \draw (10,0)  node[lab] (v3) {$3$}; 
    \draw (20,10) node[lab] (v4) {$4$}; 
    
    \draw[thick] (v1) -- (v2);
    \draw[thick] (v1) -- (v3);
    \draw[thick] (v2) -- (v3);
    \draw[thick] (v4) -- (v2); 
    \draw[thick] (v3) -- (v4);  
    
\end{tikzpicture}

\caption{The dependency graph $\Gamma$ corresponding to formula $\Phi$ defined in (\ref{eq:Phi}).}
\label{fig:depG}
\end{figure}

There is a remarkably simple (though not simple to derive) formula for the expected number of iterations in a run of algorithm PRS, which we now present.  Kolipaka and Szegedy~\cite{KS11} derived this formula as an upper bound, but it is in fact exact.  Given~$\Phi$, define $\Gamma=\Gamma(\Phi)$ to be the \emph{dependency graph} with vertex set $[m]$ (where vertex~$k$ corresponds to clause $\phi_k$) and edge relation~$\sim$ defined by $k\sim\ell$ iff $\scp(\phi_k)\cap \scp(\phi_\ell)\not=\emptyset$.  (Refer to Figure~\ref{fig:depG} for an example.)  Let $\Sigma=\{s_1,s_2,\ldots,s_m\}$ be an alphabet of $m$ symbols.  If $k\sim\ell$ then symbols $s_k$ and $s_\ell$ do not commute;  otherwise, $s_k$ and $s_\ell$ do commute, i.e., $s_ks_\ell=s_\ell s_k$.  Denote by $\mathcal{R}_\Gamma$ the set of commutation relations:
$$
\mathcal{R}_\Gamma=\{s_ks_\ell=s_\ell s_k: k\not\sim\ell\}.
$$
The set $\Sigma^*/\mathcal{R}_\Gamma$ of \emph{traces} over $\Sigma$ is the set of all words over the alphabet~$\Sigma$ quotiented by the commutation relations~$\mathcal{R}_\Gamma$.  So a trace can be thought of as word over $\Sigma$ where we regard two words as indistinguishable if one can be obtained from the other by transposing adjacent commuting symbols.  

There is an elegant expression for the generating function for traces.  Introduce indeterminates $z_1,\ldots,z_m$ corresponding to the $m$~clauses in~$\Phi$, and define 
$$P_{\Gamma}(z_1,\ldots,z_m)=\sum_{I\in\calI(\Gamma)}(-1)^{|I|}z_I,$$
where $z_I=\prod_{k\in I}z_k$ and $\calI(\Gamma)$ is the set of all independent sets in~$\Gamma$.  Note that the polynomial $P_{\Gamma}$ is the generating function of independent sets in $\Gamma$, with terms signed according to parity.  The generating function for traces $\Sigma^*/\mathcal R$ is the multivariate polynomial $T_\Gamma(z_1,\ldots,z_m)$ in which the coefficient of $z_1^{e_1}z_2^{e_2}\cdots z_m^{e_m}$ is the number of traces in which symbol $s_1$ occurs $e_1$ times, $s_2$ occurs $e_2$ times, etc. The following expression for the trace generating function is due to Cartier and Foata~\cite{CartierFoata}.  The derivation can also be found, e.g., in Knuth \cite[Thm~F]{KnuthVol4Fasc6} and Viennot \cite[Prop.~5.1]{Viennot}. 

\begin{lemma}
With $\Gamma$, $P_\Gamma$ as above, the generating function $T_\Gamma$ for traces $\Sigma^*/\mathcal{R}$ is given by $T_\Gamma(z_1,\ldots,z_m)=P_{\Gamma}(z_1,\ldots,z_m)^{-1}$.  
\end{lemma}

Take, as an example, the dependency graph $\Gamma$ from Figure~\ref{fig:depG}.  The generating function for signed independent sets in~$\Gamma$ is
$$
P_\Gamma(z_,z_2,z_3,z_4)=1 - z_1 - z_2 - z_3 - z_4 + z_1z_4,
$$
encoding the empty independent set $\emptyset$, the four singleton independent sets $\{1\}$, $\{2\}$, $\{3\}$, $\{4\}$, and the unique independent set $\{1,4\}$ of size two. Then,
\begin{align*}
P_\Gamma(z_,z_2,z_3,z_4)^{-1}&=(1-(1-P_\Gamma))^{-1}\\
&=1+(1-P_\Gamma)+(1-P_\Gamma)^2+(1-P_\Gamma)^3+\cdots\\
&=1+z_1+z_2+z_3+z_4+z_1^2+z_2^2+z_3^2+z_4^2\\
&\qquad{}+2z_1z_2+2z_1z_3+z_1z_4+2z_2z_3+2z_2z_4+2z_3z_4\\
&\qquad{}+\text{terms of degree $3$ and higher},
\end{align*}
(Note that $1-P_\Gamma$ has no constant term, so the expansion makes sense.)  Observe that the coefficient of $z_2z_3$ is 2, reflecting the fact that $s_2s_3$ and $s_3s_2$ are distinct traces, while the coefficient of $z_1z_4$ is 1, as $s_1s_4$ and $s_4s_1$ are equivalent as traces.   

The motivation for introducing traces is that they are in perfect correspondence with transcripts, where the $m$ symbols correspond to the $m$ possible kinds of resampling blocks;  specifically, symbol~$s_k$ corresponds to a block arising from resampling~$\scp(\phi_k)$.  Let $w=s_{i_1}s_{i_2}\ldots s_{i_t}$ be any word in $\Sigma^*$.  Consider the transcript that results if the algorithm PRS performs block resamplings in the order $\scp(\phi_{i_1}),\scp(\phi_{i_2}),\ldots,\scp(\phi_{i_t})$.  Now let $w'=s_{i_1'}s_{i_2'}\ldots s_{i_t'}$ be any word in $\Sigma^*$ that is equivalent to~$w$ under the commutation relations~$\mathcal R$.  It is not difficult to see that the same transcript results from the sequence of block resamplings $\scp(\phi_{i_1'}),\scp(\phi_{i_2'}),\ldots,\scp(\phi_{i_t'})$.  (Transposing the order of two adjacent commuting symbols transposes the order in which two blocks are resampled;  however, those blocks have no variables in common, so there is no change in the transcript.)  Conversely, if words $w$ and $w'$ lead to the same transcript then they must be equivalent under commutativity.  (Suppose $s_{i_1}\not=s_{i_1'}$.  Let $s_{i_h'}$ be the first occurrence of the symbol $s_{i_1}$ in $w'$.  The first $h-1$ resamplings prompted by $w'$ did not disturb the variables in $\scp(\Phi_{i_1})$.  Therefore, $s_{i_h'}$ commutes with all earlier symbols in~$w'$ and can be `bubbled' into first place.  The remaining symbols can be brought into alignment inductively.)  Transcripts are exactly the \emph{empilements [des pi\`eces]} or `heaps of pieces' of Viennot~\cite{Viennot}, who gives a beautiful pictorial explanation of the correspondence between \emph{empilements} (and hence transcripts) and traces.  See also Knuth~\cite[\S7.2.2.2]{KnuthVol4Fasc6}. 

The correspondence between traces and transcripts can be appreciated pictorially in Figure~\ref{fig:transcript}.  The depiction of the transcript is based on Viennot's \emph{empilements}.   Knuth invites us to think of each symbol as a piece in Tetris that appears from above and descends until further progress is obstructed.   The word $w=s_4s_2s_1s_3s_1s_4s_2$ specifies an order for the arriving pieces that leads to the transcript on the right of the figure.  The word $w'=s_4s_2s_1s_3s_4s_1s_2$ leads to the same transcript, since $s_1$ and $s_4$ commute.  In contrast, the word $s_4s_2s_3s_1s_1s_4s_2$ results in a different transcript, as $s_1$ and $s_3$ do not commute: the pieces corresponding to symbols $s_1$ and $s_3$ cannot pass each other.  The equivalence class $\{w,w'\}$ is a trace, since the only adjacent commuting pair of symbols is $s_1s_4$.  Traces, \emph{empilements} and transcripts are different views of the same concept.

Before analysing the runtime of algorithm PRS, let us observe that it is remarkably easy to compute the probability of observing a particular transcript such as the one in Figure~\ref{fig:transcript}.  Recall that each value in the resampling table is the result of an independent toss of a fair coin.  At time $t=0$, we have that $\omega_{4,0}=0$ and $\omega_{5,0}=1$, an event that occurs with probability~$\frac14$.  (The only way for $\phi_4(X_4,X_5)$ to be false is for $X_4$ to be 0, and $X_5$ to be~1.)  At time $t=1$ we know that $\omega_{1,0}=1$, $\omega_{3,0}=0$ and $\omega_{4,1}=1$, an event with probability $\frac18$, and so on for times $t=2,3,4,5,6$.  All these events are independent, and the probability that they all occur is $2^{-17}$.  Finally, the frontier must contain a satisfying assignment;  there are 10 satisfying assignments out of a total of 32, so the probability of observing the transcript depicted is $10\times2^{-22}$.

For $k\in[m]$, let $p_k=\Pr_\calD(\neg\phi_k)$ denote the probability that $\phi_k$ is false in the product distribution, and extend this notation to a set of clauses $S\subseteq[m]$ by letting $p_S=\prod_{k\in S}p_k$.  Then define
$$
q_S=\sum_{I\in\calI(\Gamma):I\supseteq S}(-1)^{|I\setminus S|}p_I.
$$
Note that $q_S=0$ if $S\notin\calI(\Gamma)$.  Note also that
\begin{align}
P_{\Gamma}(p_1,\ldots,p_m)&=\sum_{I\in\calI(\Gamma)}(-1)^{|I|}p_I=q_\emptyset\label{eq:q0}\\
\noalign{\noindent and}
p_kP_{\Gamma-N[k]}(p_1,\ldots,p_m)&=p_k\sum_{I\in\calI(\Gamma-N[k])}(-1)^{|I|}p_{I}\notag\\
&=\sum_{I\in\calI(\Gamma):I\ni k}(-1)^{|I|-1}p_{I}\notag\\
&=q_{\{k\}},\label{eq:qk}
\end{align}
where $\Gamma-N[k]$ denotes the graph obtained from the dependency graph~$\Gamma$ by removing the closed neighbourhood of~$k$ and incident edges.  (The \emph{closed neighbourhood} $N[k]$ of $k$ is the set containing vertex~$k$ and all its neighbours.)  

In the case of extremal instances, the quantity $q_\emptyset$ has a simple probabilistic interptetation. By the principle of inclusion-exclusion, 
\begin{align}
\Pr_\calD(\Phi)&=\Pr_\calD\bigg(\bigwedge_{k\in[m]}\phi_k\bigg)\notag\\
&=\sum_{S\subseteq[m]}(-1)^{|S|}\Pr_\calD\bigg(\bigwedge_{k\in S}\neg\phi_k\bigg)\notag\displaybreak[0]\\
&=\sum_{I\in\calI(\Gamma)}(-1)^{|I|}\prod_{k\in I}\Pr_\calD(\neg\phi_k)\label{eq:note}&\text{(see below)}\displaybreak[0]\\
&=\sum_{I\in\calI(\Gamma)}(-1)^{|I|}p_I\notag\\
&=q_\emptyset.\label{eq:q0interpretation}
\end{align} 
Equality~\eqref{eq:note} uses two facts: (a)~when $S$ is not an independent set the corresponding term is zero, by extremality, and (b)~for any independent set $I$, the events $\{\neg\phi_k:k\in I\}$ are probabilistically independent. Note, in particular, that $q_\emptyset>0$ when $\Phi$ is satisfiable.  

In the next theorem, the first sampling of the variables in some scope is regarded as a \emph{re}sampling, even though there was no previous one.

\begin{theorem}\label{thm:kresamplings} Suppose $\Phi$ is a satisfiable extremal instance.  Then the expected number of resamplings of the scope of $\phi_k$ during a run of\, $\PRS(\Phi,\calD)$ is $q_{\{k\}}/q_\emptyset$.
\end{theorem}

\begin{proof}
As noted earlier, the generating function for transcripts is $P_{\Gamma}(z_1,\ldots,z_m)^{-1}$.  We claim that the generating function for transcripts weighted according to probability of occurrence is $q_\emptyset P_{\Gamma}(p_1z_1,\ldots,p_mz_m)^{-1}$.  In other words, the probability of observing a transcript with $e_k$~resamplings of scope $\scp(\phi_k)$, for $1\leq k\leq m$, is the coefficient of $z_1^{e_1}z_2^{e_2}\cdots z_m^{e_m}$ in $q_\emptyset P_{\Gamma}(p_1z_1,\ldots,p_mz_m)^{-1}$. 
To see this, fix a transcript with $e_k$~resamplings of scope $\scp(\phi_k)$, for $1\leq k\leq m$, and consider the probability that a random resampling table will generate that transcript.  The frontier must contain a satisfying assignment, which happens with probability~$q_\emptyset$, by \eqref{eq:q0interpretation}.  Each block corresponding to a clause~$\phi_k$ must contain an assignment making $\phi_k$ false, which happens with probability~$p_k$.  All these probabilities are independent, so the overall probability of observing the transcript is $q_\emptyset\prod_{1\leq k\leq m}p_k^{e_k}$.  The claim follows.  Note that we have used that the fact that if it is \emph{possible} for a certain transcript to arise from a given resampling table it \emph{will} do so.

Note that $1-P_\Gamma(p_1,\ldots,p_m)=1-q_\emptyset\in[0,1)$, and so the power series expansion 
$$
P_\Gamma(p_1z_1,\ldots,p_mz_m)^{-1}=\sum_{i=0}^\infty\big(1-P_\Gamma(p_1z_1,\ldots,p_mz_m)\big)^i
$$ 
converges in an open neighbourhood of the point $z_1=\cdots=z_m=1$.  The expected number of resamplings of the scope of $\phi_k$ is given by 
\begin{align*}
&q_\emptyset\frac\partial{\partial z_k}P_{\Gamma}(p_1z_1,\ldots,p_mz_m)^{-1}\bigg|_{z_1=\cdots=z_m=1}\\
&\qquad\null=-q_\emptyset P_{\Gamma}(p_1z_1,\ldots,p_mz_m)^{-2}\frac\partial{\partial z_k}P_{\Gamma}(p_1z_1,\ldots,p_mz_m)\bigg|_{z_1=\cdots=z_m=1}\\
&\qquad\null=q_\emptyset P_{\Gamma}(p_1,\ldots,p_m)^{-2}p_kP_{\Gamma-N[k]}(p_1,\ldots,p_m).
\end{align*}
We use here the fact that $P_\Gamma$ is multilinear, so differentiating with respect to~$x_k$ eliminates terms corresponding to independent sets that do not include~$k$.  Using identities \eqref{eq:q0} and~\eqref{eq:qk}, we see that the expected number of times $\scp(\phi_k)$ is resampled is $q_{\{k\}}/q_\emptyset$.  
\end{proof}

We can recast the above theorem in a simple, easy to use form.

\begin{corollary}\label{cor:iterbd}
The expected number of iterations of Algorithm PRS on input $(\Phi,\calD)$ is 
$$
\Ex(\textup{\#iterations})=\frac{\Pr_\calD(\textup{Exactly one clause in $\Phi$ is false})}{\Pr_\calD(\textup{$\Phi$ is true})}.
$$
\end{corollary}

\begin{proof}
Generalising the inclusion-exclusion argument used earlier, and assuming $S\in\calI(\Gamma)$, we have
\begin{align*}
\Pr_\calD\bigg(\bigwedge_{k\in S}\neg\phi_k\wedge\bigwedge_{k\in[m]\setminus S}\phi_k\bigg)&=\sum_{S'\supseteq S}(-1)^{|S'\setminus S|}\Pr_\calD\bigg(\bigwedge_{k\in S'}\neg\phi_k\bigg)\\
&=\sum_{I\in\calI(\Gamma):I\supseteq S}(-1)^{|I\setminus S|}\prod_{k\in I}\Pr_\calD(\neg\phi_k)\\
&=\sum_{I\in\calI(\Gamma):I\supseteq S}(-1)^{|I\setminus S|}p_I=q_S.
\end{align*} 
When $S\notin\calI(\Gamma)$, the above equality continues to hold, as both side are zero.  In particular, the probability that clause $\phi_k$ is false, and all others true, is precisely $q_{\{k\}}$. The result now follows from Theorem~\ref{thm:kresamplings}.
\end{proof}

{Thanks to} Lemma~\ref{lem:confluence}, the above results are completely robust against changes in the implementation of algorithm PRS.  Thus, the next scope to be resampled can be selected by arbitrary means:  the choice can be made on the current values of variables, the past execution of the algorithm, or even externalities such as random bits or the system clock.  It is also valid to resample several blocks simultaneously, in case several clauses are violated.  
If one is interested in the expected number of individual variables resampled, this can also be accessed though
\begin{equation}\label{eq:ExNumVars}
\Ex(\text{\#variables resampled})=\sum_{k=1}^m\frac{q_{\{k\}}a_k}{q_\emptyset},
\end{equation}
where $a_k$ is the arity of $\phi_k$, for $1\leq k\leq m$.


\section{Example applications}\label{sec:apps}

One application, to sink-free orientations, will be done in detail to illustrate the methods, and the other applications merely sketched.  For ease of presentation, all examples will be unweighted, i.e, the probability distributions~$\calD_i$ are all uniform, as is the output distribution.  Incorporating weights does not require any conceptual changes.   

\subsection{Sink-free orientations of a graph}
This approach to sampling sink-free orientations of a graph was introduced by Cohn, Pemantle and Propp~\cite{CPP02}, and placed within the general framework of PRS by Guo, Jerrum and Liu~\cite[\S4.1]{GJL19}.

Suppose $G=(V,E)$ is a graph with vertex set $\{v_1,v_2,\ldots,v_m\}$ and edge set $E=\{e_1,e_2,\ldots,e_n\}$.\footnote{The roles of $n$ and $m$ are reversed relative to the usual convention in graph theory, but this is necessary to preserve consistency with the previous section.}  We wish to sample, uniformly at random, an orientation of the edges of~$G$ that has no sinks, where a \emph{sink} is a vertex $v_i$ at which all incident edges are oriented towards~$v_i$.  We assume that $G$ has at least one such sink-free orientation.  It is convenient to choose a reference orientation for the edges of~$G$ that is sink-free;  denote by $\Garrow$ the directed graph obtained from $G$ by giving the edges of~$G$ this reference orientation.   

\begin{figure}[t]
\centering
\begin{tikzpicture}[xscale=0.15, yscale=0.15, inner sep=2pt, >=stealth]

    \draw (0,10)  node[lab] (v1) {$v_1$}; 
    \draw (10,20) node[lab] (v2) {$v_2$};
    \draw (10,0)  node[lab] (v3) {$v_3$}; 
    \draw (20,10) node[lab] (v4) {$v_4$}; 
    
    \draw[->,thick] (v1) -- (v2);
    \draw[->,thick] (v1) -- (v3);
    \draw[->,thick] (v2) -- (v3);
    \draw[->,thick] (v4) -- (v2); 
    \draw[->,thick] (v3) -- (v4);  
    
    \draw (4,16)  node[text width = 6mm] () {$e_1$}; 
    \draw (4,4)   node[text width = 6mm] () {$e_2$};
    \draw (13,10) node[text width = 6mm] () {$e_3$};
    \draw (18,16) node[text width = 6mm] () {$e_4$}; 
    \draw (18,4)  node[text width = 6mm] () {$e_5$};

  \end{tikzpicture}

\caption{A sample graph $\Garrow$ incorporating a reference orientation}
\label{fig:sinkfreeex}
\end{figure}

To fit the pattern of PRS, we introduce Boolean variables $X_1,X_2,\ldots,X_n$ and associate variable~$X_i$ to edge~$e_i$, for $1\leq i\leq n$.  These variables will be used to encode orientations of the edges of~$G$.  The variable $X_i$ is to be interpreted as follows:  if $X_i=0$ then the edge $e_i$ is oriented against the reference orientation (of $e_i$ in $\Garrow$) and if $X_i=1$ then $e_i$ is oriented with the reference orientation.  Next, introduce clauses~$\{\phi_k\}$ to encode the event that vertex $v_k$ is not a sink.  So the scope of $\Phi_k$ is the set $\scp(\Phi_k)=\{i:\text{$e_i$ is incident at $v_k$}\}$, and the clause $\phi_k$ asserts that at least one edge incident at vertex~$v_k$ is oriented away from~$v_k$.  By way of example, consider the graph $G$ in Figure~\ref{fig:sinkfreeex}, which has been assigned a reference orientation to give a sink-free directed graph~$\Garrow$.  The condition that vertex~$v_2$, for example, is not a sink is asserted by the clause $\phi_2(X_1,X_3,X_4)=\neg X_1\vee X_3\vee\neg X_4$.  Then $\Phi$ is the formula
$$
\Phi(\bfX)=(X_1\vee X_2)\wedge(\neg X_1\vee X_3\vee\neg X_4)\wedge(\neg X_2\vee\neg X_3\vee X_5)\wedge(X_4\vee \neg X_5)
$$
that we encountered already in the previous section.  

We observed earlier that $\Phi$ is an extremal instance.  This is true in general for sink-free orientations.  If we have indices $1\leq k<\ell\leq m$ such that $\scp(\phi_k)\cap\scp(\phi_\ell)\not=\emptyset$ then necessarily vertices $v_k$ and~$v_\ell$ are adjacent.  But then it is impossible for $v_k$ and $v_\ell$ to both be sinks, and hence $\phi_k\vee\phi_\ell$ must hold.  So Theorem~\ref{thm:PRScorrect} immediately assures us that PRS will produce a uniform random sink-free orientation with probability~1.  But is the expected running time polynomial in $n$ and~$m$?  In order to apply Corollary~\ref{cor:iterbd} we need to bound the ratio $q_{\{k\}}/q_\emptyset$.  Although we don't have a handle on $q_{\{k\}}$ and $q_\emptyset$ --- and, in a sense, $q_\emptyset$ is a quantity we would like to compute --- we can bound the ratio by defining an appropriate mapping from orientations with exactly one sink to those with none.

Introduce a function $f$ from $\{1,\ldots,m\}$ to itself that is consistent with the reference orientation, that is to say, $(v_k,v_{f(k)})$ is a (directed) edge in $\Garrow$ for all $1\leq k\leq m$.  This is possible because the reference orientation is sink-free.  To each orientation of $G$ that has a single sink at $v_k$ we associate a sink-free orientation as follows.  Let $e_i$ be the edge $(v_k,v_{f(k)})$.  Reverse the orientation of~$e_i$, i.e., set $X_i$, which was previously 0, to~1.  Vertex~$v_k$ is no longer a sink, but $v_{f(k)}$ may have become one.  If $v_{f(k)}$ is not a sink than halt.  Otherwise reverse the orientation of the edge $(v_{f(k)},v_{f^2(k)})$, and continue.  This process must terminate.  For suppose not.  Let $t$ be the first instant at which we revisit a vertex, i.e., such that $f^t(k)=f^s(k)$ for some $0\leq s<t$.  The edge $(v_{f^s(k)},v_{f^{s+1}(k)})$ is directed away from $v_{f^s(k)}$, and hence vertex $v_{f^s(k)}$ is not a sink, a contradiction.  (It is important to note that we leave and revisit vertex~$v_s$ via different edges.) 

The edges that were flipped in the above construction form a path $v_k=v_{f^0(k)},\allowbreak v_{f^1(k)},\ldots,\allowbreak v_{f^\ell(k)}$.  We may undo the construction provided we know $f^0(k)=k$ and $f^\ell(k)$.  It follows that the number of orientations with a single sink exceeds the number of sink free orientations by a factor at most~$m^2$.  So by Corollary~\ref{cor:iterbd} the expected number of iterations in a run of PRS --- in this case the number of sinks that are `popped' --- is bounded above by $m^2=|V(G)|^2$.  We may also bound the number times the orientations of individual edges are flipped.  Fix a vertex $v_k$.  We saw above how to repair an orientation with a single sink at $v_k$.  To undo this repair, we just need to specify the index $f^\ell(k)$.  Thus the number of orientations with a single sink at~$v_k$ exceeds the number of sink-free orientations by a factor~$m$.    Referring to \eqref{eq:ExNumVars}, we have $q_{\{k\}}/q_\emptyset\leq m$ and $a_k=\deg(v_k)$, the degree of vertex~$v_k$. Thus the expected number of edge orientation reversals is $\sum_{k=1}^m q_{\{k\}}a_k/q_\emptyset\leq \sum_{k=1}^m m\deg(v_k)\leq2mn$.  So the expected number of orientation reversals is a most $2|V(G)|\,|E(G)|$.  All this is in agreement with~\cite{CPP02}.  

Surprisingly, the upper bound on edge-reversals can be tightened further to $|E(G)|+|V(G)|^2$:  see Guo and He~\cite{GH18}.  Note that, the runtime analysis critically used the assumption that coin tosses are unbiased, so that either orientation of an edge is equally likely.  (A simple counterexample shows that this assumption is necessary.) In contrast, correctness of the algorithm extends to asymmetric orientation probabilities.

\subsection{Spanning trees of a graph}

The Cycle-popping algorithm is an approach to uniformly sampling spanning trees in a graph, introduced by Propp and Wilson~\cite{PW98}.  Suppose $G$ is a graph with vertex set $V=\{v_0,v_1,\ldots,v_n\}$ and edge set~$E$.  Instead of sampling spanning trees in~$G$ we will instead sample spanning (in-)arborescences\footnote{That is, directed spanning trees with edges directed towards a root vertex.} rooted at~$v_0$, which is of course equivalent.  

For each $1\leq i\leq n$, define $D_i=\{j:\{v_i,v_j\}\in E\}$, and make $D_i$ into a probability space~$\calD_i$ by equipping it with the uniform distribution.  Introduce random variables $X_1,\ldots,X_n$ distributed as $\calD_1,\ldots,\calD_n$.  These variables indicate, for each $1\leq i\leq n$, a possible exit from vertex~$v_i$.  For each simple (oriented) cycle $C=(v_{i_0},v_{i_1},\ldots,v_{i_{\ell-1}},v_{i_{\ell}}=v_{i_0})$ define the predicate $\phi_C$ by 
$$
\phi_C=\neg(X_{i_0}=i_1\wedge X_{i_1}=i_2\wedge\cdots\wedge X_{i_{\ell-1}}=i_\ell),
$$
and the formula $\Phi$ by $\Phi=\bigwedge_C\phi_C$, where the conjunction is over all oriented cycles in~$G$.  (In this context, `simple' is taken to mean `containing no repeated vertices';  thus we regard the 2-cycle $(v_{i_0},v_{i_1},v_{i_0})$ as simple.)  The intended interpretation of the event $X_i=j$ is that vertex $v_j$ is the ancestor of vertex $v_i$ in the arborescence.  The formula $\Phi(\bfX)$ asserts that the ancestor relation is consistent (has no cycles) and hence that $\bfX$ encodes a spanning arborescence rooted at~$v_0$.  

Consider two clauses $\phi_C$ and $\phi_{C'}$ corresponding to distinct cycles $C$ and~$C'$.  If $\scp(\phi_C)\cap\scp(\phi_{C'})\not=\emptyset$ then $C$ and~$C'$ must have a vertex in common.  Select a vertex~$v_i$ that is common to $C$ and~$C'$ with the additional property that the successor to~$v_i$ in cycle~$C$ is not equal to the successor to~$v_i$ in cycle~$C'$.  Let $v_j$ be the successor to~$v_i$ in~$C$ and $v_{j'}$ be the successor in~$C'$.  It is clear that $X_i=j$ and $X_i=j'$ cannot both be true, and hence $\phi_C$ and $\phi_{C'}$ cannot both be false.  Therefore $\Phi$ is extremal.  

As, in the previous example, we need to estimate the ratio between aborescences and `near-arborescences' that contain a single cycle.  (A near-arborescence has two components:  a spanning arborescence on some subset $S\ni v_0$ of the vertices~$V$ of~$G$, rooted at~$v_0$, and a unicyclic directed subgraph spanning $V\setminus S$.)  As before, by considering a suitable mapping from near-arborescences to arboresecences, it can be shown that the number of the former is at most $|V(G)|\,|E(G)|$ times the number of the latter.  Thus, by Corollary~\ref{cor:iterbd}, the number of iterations made by PRS is at most $|V(G)|\,|E(G)|$.  A more refined analysis, due to Guo and He~\cite[Thm~15]{GH18}, shows that the total number of \emph{variable} updates is bounded by essentially the same expression.  

\subsection{Root-connected subgraphs}

This `cluster-popping' algorithm was proposed by Gorodezky and Pak~\cite{GP14}, who conjectured it to be efficient on a certain class of directed graphs.  The conjecture was resolved affirmatively by Guo and Jerrum~\cite{GJ19}.

Suppose $G=(V,A)$ is a directed graph with a distinguished root vertex $r\in V$.  A~spanning subgraph $(V,S)$ of~$G$ is said to be \emph{root-connected} if, for every vertex $v\in V$, there is a directed path in $(V,S)$ from~$v$ to~$r$.  Our task is to sample, uniformly at random, a root-connected subgraph of $G$.  As usual, we restrict our attention to the unweighted version.  However, as we shall note later, the weighted version is of interest, owing to its connection to a network reliability problem. 

A subgraph $(V,S)$ may be encoded by variables $\bfX=(X_e:e\in A)$ taking values in $\{0,1\}$.  The interpretation of $X_e=1$ is that $e\in S$.  For an arc $e\in A$, denote by $e^-$ and $e^+$ the start and end vertex of~$e$.  A \emph{cluster} in $(V,S)$ is a set $\emptyset\subset C\subseteq V\setminus\{r\}$ of vertices with the property that no edge $e\in S$ exists with $e^-\in C$ and $e^+\in V\setminus C$.   The property `$C$ is a cluster' can be expressed formally by the predicate $\psi_C=\bigwedge_{e\in\mathrm{cut}(C)}\neg X_e$, where $\mathrm{cut}(C)=\{e\in A:e^-\in C\text{ and }e^+\in V\setminus C\}$.  If the subgraph $(V,S)$ has a cluster~$C$ then it is clear that no vertex in~$C$ can reach $r$, via a directed path in $(V,S)$, and hence $(V,S)$ is not root-connected.  The converse is also true:  Suppose $(V,S)$ is not root-connected, and let $v$ be some vertex from which the root~$r$ is not reachable.  Let $C$ be the set of all vertices reachable from $v$.  Then $C$ is a cluster in $(V,S)$.  

This observation suggests that we should define 
\begin{equation}\label{eq:PhiCluster}
\Phi=\bigwedge_{\emptyset\subset C\subseteq V\setminus\{r\}}\phi_C,
\end{equation}
where $\phi_C=\neg\psi_C$.  The formula $\Phi$ denies the existence of a cluster in the subgraph encoded by~$\bfX$, and hence correctly expresses the property of being root-connected.  The catch is that $\Phi$ is not in general extremal.  It is perfectly conceivable that two clusters $C,C'$ exist that have nonempty intersection $C\cap C'\not=\emptyset$.  In that case, we might have $\scp(\phi_C)\cap\scp(\phi_{C'})\not=\emptyset$ and yet $\phi_C$ and $\phi_{C'}$ are both false.  The solution is to make the predicates $\phi_C$ less demanding, while preserving the semantics of~$\Phi$.  We say that the cluster $C$ is \emph{minimal} if it contains no cluster $C'$ with $C'\subset C$.  Then we define $\phi_C$ to be true if $C$ is not a minimal cluster.  Formally, 
$$
\phi_C=\neg\Big[\psi_C\wedge \bigwedge_{\emptyset\subset C'\subset C}\neg\psi_{C'}\Big]=\neg\psi_C\vee\bigvee_{\emptyset\subset C'\subset C}\psi_{C'}.
$$
Then define $\Phi$ as in \eqref{eq:PhiCluster}.  We claim that $\Phi$ still expresses the condition that $\bfX$ encodes a root-connected subgraph $(V,S)$.  If $(V,S)$ is root-connected, then no cluster exists and hence $\phi_C$ is satisfied for all $\emptyset\subset C\subseteq V\setminus\{r\}$.  Conversely, suppose that $(V,S)$ is not root-connected.  Then there is at least one cluster, and hence at least one minimal cluster~$C$.  For this cluster, $\phi_C$ is contradicted, and hence $\Phi$ is false.     

Although the meaning of $\Phi$ is unchanged, the formula is now extremal.  First note that, for all subsets~$C$, 
$$
\scp(\phi_C)=\{X_e:e^-\in C\}.
$$    
So if $\phi_C$ and $\phi_{C'}$ are any two distinct clauses with $\scp(\phi_C)\cap\scp(\phi_{C'})\not=\emptyset$, we must have $C\cap C'\not=\emptyset$.  If $C$ and $C'$ are both clusters then $C\cap C'$ must also be a cluster.  Therefore, $C$ and~$C'$ cannot both be minimal clusters.  It follows that at least one of $\phi_C$ or $\phi_{C'}$ must hold.  This deals with correctness of PRS in this context.  

Unfortunately, PRS does not have expected polynomial runtime on general instances~$G$, as can be appreciated by considering a counterexample presented by Gorodezky and Pak~\cite{GP14}.  However, those same authors conjectured that the runtime is polynomial when the graph~$G$ is `bidirected', i.e., an edge exists from vertex $u$ to~$v$ in~$G$ if and only if an edge exists from $v$ to $u$.  This special case is of interest, since root-connected subgraphs of a bidirected graph $G$ correspond (via a constantly many-one relation) to spanning connected subgraphs of the undirected version of~$G$.  Thus, cluster popping provides a efficient approach to sampling connected spanning subgraphs of a graph.   

The conjecture of Gorodezky and Pak may be verified using Corollary~\ref{cor:iterbd}.  Again the argument involves a mapping from subgraphs with exactly one minimal cluster to root-connected subgraphs.  The combinatorial details of this mapping and its analysis, which are more involved in this case that the previous ones, are given by Guo and Jerrum~\cite{GJ19}.  The resulting upper bound on the expected number of variable resamplings is $|V(G)|\,|E(G)|^2$, which can be improved to $|V(G)|\,|E(G)|$ by a more refined analysis~\cite{GH18}.  For a short while, PRS provided the only known attack on sampling connected spanning subgraphs of a general undirected graph, and its weighted version, undirected all-terminal reliability.  However the same problem (in a more general setting) has since been solved by Markov chain simulation by Anari, Liu, Oveis Gharan and Vinzant~\cite{ALOV}.

\subsection{Bases of bicircular matroids}
Another application of PRS is to sampling bases of a bicircular matroid.  The algorithm was first presented in a slightly different guise by Kassel and Kenyon~\cite{KK17}.  Suppose $G=(V,E)$ is an undirected graph.  The \emph{bicircular matroid} associated with~$G$ has~$E$ as its ground set. The bases of the matroid are all spanning subgraphs of $G$ in which every connected component is unicyclic;  equivalently, every connected component has the same number of edges as it has vertices.  The sampling algorithm may be derived methodically using PRS{}.  The application has similarities with the cycle-popping algorithm described above in the context of sampling spanning trees.  

As with cycle popping, variables are introduced that encode a function~$g$ from~$V$ to itself that respects the edges of~$G$.  (This is a slight deviation from the spanning trees case, where the function was from $V\setminus\{r\}$ to~$V$.)  The spanning subgraph $(V,S)$ defined by $S=\{\{v,g(v)\}:v\in V\}$ is very like a basis of the bicircular matroid, with two caveats.  First, we want to rule out cycles of length~2 --- that is, situations in which $g(g(v))=v$ for some $v\in V$ --- as such functions~$g$ do not correspond to valid bases.  Second, each basis with $c$~connected components corresponds to $2^c$ distinct functions, as each cycle may be traced in either orientation.   
  
To deal with these two objections, we specify a preferred orientation for every cycle in~$G$.  Our formula $\Phi$ includes a clause $\phi_C$, for each potential cycle~$C$ that either (a)~has length two, or (b)~is oriented in the in the wrong sense.  In each case, $\phi_C$ asserts that $C$ does not occur.  It is easy to check that $\Phi$ is extremal.  The expected number of resamplings (either of clauses or individual variables) is $O(|V(G)|^2)$.  Details are given by Guo and Jerrum~\cite{GJ18c}.  

\subsection{Notes}
The examples listed above are not the only known applications of PRS, but they are the only non-trivial ones for which polynomial-time running time bounds are known.  At least, they are the one ones I am aware of.

One tempting potential application is to sampling strong orientations of an undirected graph.  An orientation of the edges of an undirected graph $G$ is \emph{strong} if there is a directed path from every vertex of $G$ to every other.  If $G$ is connected, strong orientations coincide with `totally cyclic orientations'.  The number of total cyclic orientations of a graph $G$ is an evaluation of the Tutte polynomial (at the point $(0,2)$).  It is known that counting totally cyclic (and hence strong orientations) is $\numP$-complete~\cite{JVW}.  However, the computational complexity of approximately counting or uniformly sampling totally cyclic orientations is unknown.  

The cluster-popping algorithm for root-connected subgraphs is easily adapted to strong orientations.  For a set $\emptyset\subset S\subset V=V(G)$ of vertices of~$G$, we say that $S$ is \emph{cluster} if all edges between $S$ and $V\setminus S$ are directed into~$S$.  (The crucial difference with the root-connected case is that there is no distinguished root vertex~$r$ that is excluded from all clusters.)  We say that a cluster is \emph{minimal} if it is minimal with respect to inclusion.  As usual, define $\Phi=\bigwedge_{\emptyset\subset S\subset V}\phi_S$, where the formula~$\phi_S$ expresses the condition that $S$ is not a minimal cluster.  It may be verified that $\Phi$ is extremal, and hence that PRS produces a uniform random strong orientation (assuming that the $G$ has one, which happens exactly when the graph $G$ is bridgeless).  Unfortunately, the expected runtime may be exponential, as can be appreciated by considering the ladder graph~$L_n$ on $2n$ vertices.  (The ladder graph can be viewed as a $n\times 2$ rectangular piece of the square lattice, or as the cartesian product $P_n\times P_2$ of a path on $n$~vertices and a path on $2$ vertices.)  If $v_k$ is one of the degree-2 corner vertices then the ratio $q_{\{k\}}/q_\emptyset$ from Theorem~\ref{thm:kresamplings} is exponential in~$n$.  (By induction on~$n$, the number of strong orientations is $2\times3^{n-2}$, whereas the number of orientations with a unique minimal cluster $\{v_k\}$ is at least $4^{n-1}$.)  The fact that Theorem~\ref{thm:kresamplings} gives an exact result and not just an upper bound comes in useful here, as it enables us to deduce a \emph{lower} bound on the running time of PRS.

\section{Non-extremal instances}\label{sec:nonextremal}

In an extremal instance, no two clauses that share variables can be simultaneously false.  We have seen that this leads to uniform outputs from PRS.  It transpires that we can get away with a little less than this.  
\begin{definition}\label{def:quasi-extremal}
We say that the formula $\Phi=\phi_1\wedge\cdots\wedge\phi_m$ is \emph{quasi-extremal} if the following holds, for all $k,\ell\in[m]$ and assignments $\bfX$ and $\bfX'$: if $\neg\phi_k(\bfX)\wedge\neg\phi_\ell(\bfX)$ and it is possible to get from~$\bfX$ to~$\bfX'$ by resampling variables in the scope of~$\phi_\ell$, then $\neg\phi_{k'}(\bfX')$ for some~$k'$ with $\scp(\phi_{k'})\supseteq\scp(\phi_k)$.  
\end{definition}
Note that an extremal instance satisfies the above definition with $k'=k$, so the qualifier `quasi-extremal' is a weakening of `extremal'.  The additional flexibility allows PRS to be applied to a significantly wider class of examples.

The algorithm is exactly as before except for one change.  Although we have considerable flexibility in the order in which to resample (the scopes of) clauses, we no longer have complete freedom.

\begin{algorithm}
\caption{Partial Rejection Sampling with limited nondeterminism}\label{alg:PRS-nonextremal}
\begin{algorithmic}
\STATE $\PRS(\Phi,\calD)$ 
\COMMENT {$\Phi$ is a formula on variable set $\bfX$}
\STATE {Sample $\bfX$ from the product distribution $\calD_1\times\cdots\times\calD_n$}
\WHILE {$\neg\Phi(\bfX)$}
\STATE {$N:=\{\ell:\neg\phi_\ell(X_{\scp(\phi_\ell)})\}$}
\STATE {Choose $k\in N$ deterministically, \emph{based only on the set $N$ itself}}
\STATE {Resample all variables in $\scp(\phi_k)$}
\ENDWHILE
\end{algorithmic}
\end{algorithm}

\begin{theorem}\label{thm:PRScorrect-nonextremal}
Suppose $\Phi$ is a quasi-extremal satisfiable instance.  Then $\PRS(\Phi,\calD)$ terminates with probability~1.  On termination, $\bfX$ is a realisation of a random variable from the distribution~$\calD_\Phi$.
\end{theorem}

\begin{proof}
Termination with probability 1 can be argued exactly as in the proof of Theorem~\ref{thm:PRScorrect}.

For correctness, we set up the resampling table as in the proof of Theorem~\ref{thm:PRScorrect}. 
As before, fix a resampling table $T$, and run $\PRS$ on $T$ to obtain a transcript.  Since the algorithm has terminated, we know that the frontier contains a satisfying assignment.  Create a new resampling table~$T'$ by replacing the values in the frontier by some other satisfying assignment.  Now run the algorithm on the new resampling table~$T'$.  We claim that this second run correctly outputs the planted satisfying assignment.

If, in both runs of the algorithm, the same clause~$\phi_k$ is selected in every iteration then the output indeed will be correct.  So assume that in some iteration different clauses are selected in the two runs.  For this to occur, the set~$N$ must differ between the two runs.  Consider the first iteration on which this occurs, and suppose~$\phi_k$ is true in one run and false in the other.  As before, let $S=\scp(\phi_k)$ and partition~$S$ as $S=S_I\cup S_F$, where variables~$S_I$ (respectively, $S_F$) take values from the interior (respectively, frontier) of the resampling table.  Note that $S_I\not=\emptyset$ (otherwise $\phi_k$ would be true in both runs) and $S_F\not=\emptyset$ (otherwise $\phi_k$ would have the same truth value in both runs).  

There are two cases, both of which lead to a contradiction.  Suppose first that $\phi_k$ is false in the $T$-run (and incidentally true in the $T'$-run, thought this is not relevant to the argument).  Allow the $T$-run to continue.  The algorithm does not resample $\scp(\phi_k)$ itself, since that action would take it past the final frontier of the table.  If it resamples $\scp(\phi_\ell)$ for some $\ell\not=k$ then, since $\Phi$ is quasi-extremal, this action would leave behind a clause $\phi_{k'}$ with $\phi_{k'}$ false and $\scp(\phi_{k'})\supseteq\scp(\phi_k)$.  Arguing as before, the algorithm does not resample $\scp(\phi_{k'})$ so, by induction, at least one clause will always be false for the remainder of the run.  Thus, we can never make all clauses of~$\Phi$ true, which contradicts the fact that the transcript is finite.

The second and final case has $\phi_k$ false in the $T'$-run (and incidentally true in the $T$-run).  Up to this point, the two runs have made exactly the same choices of scopes to resample.  Now imagine that all the resampling steps in the $T$-run are faithfully mirrored in the $T'$ run.  We have deviated from the deterministic choice rule of the algorithm, but all resampling steps are legal, in the sense that we always resample scopes $\scp(\phi_\ell)$ for which $\phi_\ell$ is currently false.  The reason for this is exactly as in the proof of Theorem~\ref{thm:PRScorrect}:  briefly, that we never resample variables in the frontier and the variables sampled from the interior have the same values in both runs.  The $T'$-run finishes with the same transcript as the $T$-run, but with a different satisfying assignment in the frontier. In particular all clauses of $\Phi$ are satisfied.  On the other hand, we may argue as follows.  Since $\scp(\phi_k)$ is never resampled again in the $T$-run, it is also never resampled in the $T'$-run.  Also, as in the first case, by resampling $\scp(\phi_\ell)$ for $\ell$ with $\ell\not=k$, we must leave behind a clause $\phi_{k'}$ with $\phi_{k'}$ false and $\scp(\phi_{k'})\supseteq\scp(\phi_k)$.  Arguing again by induction, we can never make all clauses of~$\Phi$ true, which is a contradiction.

Summarising, the final transcript remains unchanged if one satisfying assignment is substituted for another in the final frontier.  Thus, conditioned on the transcript, each satisfying assignment $\bfX=(b_1,\ldots,b_n)$ of~$\Phi$ occurs with  probability proportional to $\calD_1(b_1)\*\calD_2(b_2)\cdots\calD_n(b_n)$.   So, at termination, $\bfX$ is distributed as~$\calD_\Phi$. 
\end{proof}

\subsection{Example: Independent sets (the hard-core gas model)}

Suppose wish to sample independent sets in a graph~$G$.  Introduce variables $\bfX=(X_v:v\in V(G))$ taking values in $\{0,1\}$ to encode potential independent sets in~$G$.  The interpretation of $X_v=1$ (respectively $X_v=0$) is that vertex $v$ is in (respectively not in) the independent set.  In our product distribution we assume $\Pr(X_v=1)=\lambda/(1+\lambda)$ for all vertices of~$G$, for some positive `activity' $\lambda$.  (It is not essential that the activity is constant over vertices, but it slightly simplifies the exposition.)   We wish to sample from the conditional (Gibbs) distribution given that $\bfX$ encodes an independent set.

The natural formula expressing that $\bfX$ encodes an independent set is $$\Phi'(\bfX)=\bigwedge_{\{u,v\}\in E(G)}\phi_{\{u,v\}}$$ where $\phi_{\{u,v\}}=\neg(X_u\wedge X_v)$.  However $\Phi'$ is not extremal.  Following the example provided by cluster popping for root connected subgraphs, we try to re-express $\Phi'$ as a semantically equivalent extremal formula.

Let $S\subseteq V$ be a subset of at least two vertices that induces a connected subgraph $G[S]$ of~$G$.  Denote by 
$$\partial S=\{v:v\notin S\text{ and }\exists u\in S\text{ such that }\{u,v\}\in E(G)\}$$ 
the boundary of~$S$, containing all vertices outside of~$S$ that are adjacent to some vertex in~$S$.  We say that $S$ is a \emph{cluster} (relative to the assignment~$\bfX$) if $v\in S$ implies $X_v=1$ and $v\in\partial S$ implies $X_v=0$.  Refer to Figure~\ref{fig:clusters}, where solid (respectively, open) vertices $v$ are ones where $X_v=1$ (respectively, $X_v=0$).  For each $S$ of the above form, we introduce a clause~$\phi_S$ that asserts that $S$ is not a cluster.  Let $\Phi=\bigwedge_S \phi_S$ where $S$ ranges over all vertex subsets of size at least two that induce a connected subgraph.  It is clear that $\Phi(\bfX)$ asserts that $\bfX$ encodes an independent set.  

\begin{figure}[t]
\begin{tikzpicture}[scale=0.5]
    \foreach \i in {-5,...,5} {
        \foreach \j in {-3,...,4} {
            \draw (\i,\j) node [vertex] (\i X\j) {};
        }
    }
    \draw (-2,0) node [vertex, fill=black] () {};
    \draw (-2,1) node [vertex,fill=black] () {};
    \draw (-1,-1) node [vertex,fill=black] () {};
    \draw (-1,0) node [vertex,fill=black] () {};
    \draw (-1,1) node [vertex,fill=black] () {};
    \draw (1,0) node [vertex,fill=black] () {};
    \draw (1,1) node [vertex,fill=black] () {};
    \draw (1,2) node [vertex,fill=black] () {};
    \draw (2,1) node [vertex,fill=black] () {};
    \def\x{0.4}
    \draw[rounded corners=5pt] (-2-\x,2+\x) -- (-1+\x,2+\x) -- (-1+\x,1+\x) -- (0+\x,1+\x) -- (0+\x,-1-\x) -- (-1+\x,-1-\x) -- (-1+\x,-2-\x) -- (-1-\x,-2-\x) -- (-1-\x,-1-\x) -- (-2-\x,-1-\x) -- (-2-\x,0-\x) -- (-3-\x,0-\x) -- (-3-\x,1+\x) -- (-2-\x,1+\x) -- cycle;
    \draw[rounded corners=5pt] (1+\x,3+\x) -- (1+\x,2+\x) -- (2+\x,2+\x) -- (2+\x,1+\x) -- (3+\x,1+\x) -- (3+\x,1-\x) -- (2+\x,1-\x) -- (2+\x,0-\x) -- (1+\x,0-\x) -- (1
    +\x,-1-\x) -- (1-\x,-1-\x) -- (1-\x,0-\x) -- (0-\x,0-\x) -- (0-\x,2+\x) -- (1-\x,2+\x) -- (1-\x,3+\x)-- cycle;
    \foreach \i in {-5,...,5} {
        \foreach[evaluate={\jj=int(\j+1)}] \j in {-3,...,3} {
            \draw [thin] (\i X\j) -- (\i X\jj);
        }
    }
    \foreach[evaluate={\ii=int(\i+1)}] \i in {-5,...,4} {
        \foreach \j in {-3,...,4} {
            \draw [thin] (\i X\j) -- (\ii X\j);
        }
    }
\end{tikzpicture}
\caption{A pair of  clusters in  $\Zset^2$ with overlapping boundaries}
\label{fig:clusters}
\end{figure}

Unfortunately, a moment's reflection reveals that $\Phi$ is also not extremal.  Denote by~$\Sbar$ the set $\Sbar=S\cup\partial S$ and note that $\scp(\phi_S)=\Sbar$.  It is possible to have clusters $S$ and $S'$ with $\Sbar\cap\overline{S'}\not=\emptyset$, in which case $\scp(\phi_S)\cap\scp(\phi_{S'})\not=\emptyset$, and yet $\phi_S\vee\phi_{S'}$ is false.  For example, the path on five vertices $\{1,2,3,4,5\}$ with $X_3=0$ and $X_1=X_2=X_4=X_5=1$ has clusters $\{1,2\}$ and $\{4,5\}$, and $\partial\{1,2\}\cap\partial\{4,5\}=\{3\}\not=\emptyset$.  

However, it is straightforward to verify that $\Phi$ is quasi-extremal.   Suppose $\phi_S$ and $\phi_{S'}$ are simultaneously false.  It is easy to see that 
$$\Sbar\cap S'=(S\cup\partial S)\cap S'=S\cap S'=\emptyset,$$
and similarly that $S\cap\overline{S'}=\emptyset$.  If $\Sbar\cap \overline{S'}=\emptyset$ then $\scp(\phi_S)\cap\scp(\phi_{S'})=\emptyset$, and Definition~\ref{def:quasi-extremal} is satisfied with $k'=k$.  Otherwise, we are in the case 
$$\overline S\cap \overline{S'}=(S\cup\partial S)\cap(S'\cup\partial S')=\partial S\cap\partial S'\not=\emptyset.$$   
Resampling $\scp(\phi_{S'})$ can make $\phi_S$ true, but only at the expense of making some $\phi_{S''}$ with $S''\supset S$ false, since no variable in $S$ is resampled.  Again Definition~\ref{def:quasi-extremal} is satisfied, but now with $k'\not=k$.  In Figure~\ref{fig:clusters}, resampling the left cluster (with boundary) may increase the right cluster but cannot decrease it.

Specialising the generic PRS algorithm to this example, we obtain the following algorithm for sampling independent sets, which is a slight variant of one first described by Guo, Jerrum and Liu~\cite{GJL19}. 

\begin{algorithm}
\begin{algorithmic}
\STATE $\PRSforIS(G,\lambda)$ 
\COMMENT {$G$ is an undirected graph, and $\lambda$ a positive real number}
\STATE {Sample $\bfX$ from the product distribution $\mathrm{Bernoulli}(\lambda/(1+\lambda))^n$}
\WHILE {$\bfX$ does not encode an independent set}
\STATE {Choose a cluster $S$ using a valid rule}  
\STATE {Resample all variables $\{X_v:v\in\Sbar\}$}
\ENDWHILE
\end{algorithmic}
\caption{Partial Rejection Sampling for independent sets}\label{alg:PRSforIS}
\end{algorithm}

\begin{lemma}\label{lem:IScorrect}
$\PRSforIS(G,\lambda)$ terminates with probability~1.  On termination, $\bfX$ is a realisation of a random variable from the Gibbs distribution for independent sets in $G$ with activity~$\lambda$.
\end{lemma}

\begin{proof}
Follows immediately from Theorem~\ref{thm:PRScorrect-nonextremal}.
\end{proof}

\subsection{Runtime Analysis}

Sampling independent sets is in general an NP-hard problem~\cite[Thm~4]{LubyVigoda}, so we need to make some assumption about the graph~$G$ and activity~$\lambda$.  Our goal is to find $\lambda_\Delta>0$ such that PRS terminates rapidly, for all $\lambda<\lambda_\Delta$ and all graphs~$G$ of maximum degree~$\Delta$.  

We take as our starting point the runtime analysis for extremal instances from Section~\ref{sec:extremal}.  One problem extending this analysis to the non-extremal situation is that the proof of Theorem~\ref{thm:kresamplings} fails.  The reason for this is that the interpretation of $q_\emptyset$ as the probability $\Pr_\calD(\Phi)$ that $\Phi$ is satisfied is no longer valid.  It transpires that this problem can be avoided by using a different line of proof.   Kolipaka and Szegedy~\cite[Thm 4]{KS11} show that the number of resamplings of the scope of~$\phi_k$ is bounded above by $q_{\{k\}}/q_\emptyset$, provided the point $(p_1,\ldots,p_m)$ lies within a certain region.  (Refer to the preamble to Theorem~\ref{thm:kresamplings} for notation.) This region was identified by Shearer~\cite{Shearer85} as the theoretical limit of validity of the Lov\'asz Local Lemma, even in the non-algorithmic setting.  Although elegant, it is difficult to use this result directly:  testing membership in the Shearer region in specific examples is challenging, as is computing $q_{\{k\}}$ and $q_\emptyset$, which no longer have simple combinatorial interpretations.  Fortunately, there are several weaker conditions that can be feasibly tested.  

Just as we weakened the definition of extremal to quasi-extremal, we can weaken the concept to dependency graph or relation to a lopsided dependency (`lopsidependency') graph~\cite[\S6]{MT10}.
\begin{definition}\label{def:lopsided}Given a satisfiable instance $\Phi=\phi_1\wedge\cdots\wedge\phi_m$, let $k,\ell\in[m]$ be arbitrary.  Suppose there is a resampling table relative to which it is possible to resample $\scp(\phi_k)$ and then immediately resample $\scp(\phi_\ell)$ but it is not possible to perform these operations in the reverse order (either because $\phi_\ell$ is true initially, or because $\phi_k$ is true after $\scp(\phi_\ell)$ has been resampled).  Then we write $k\sim \ell$ and say that $k$ and $\ell$ are \emph{lopsidedly dependent}. The graph $([m],\sim)$ is the \emph{lopsided dependency graph} of\/ $\Phi$.
\end{definition}
Note also that the lopsided dependency graph is a subgraph, in general strict, of the usual dependency graph.  In the independent set example it is easy to characterise the lopsided dependency graph: specifically, $S\not\sim S'$ iff $S\cap \overline{S'}=\emptyset$ and $\Sbar\cap S'=\emptyset$.  To see this, consider two clauses $\phi_S$ and $\phi_{S'}$ with $S\not\sim S'$.  Suppose that the resampling block for $\scp(\phi_{S'})$ lies immediately above that for $\scp(\phi_{S})$ in the resampling table.  We claim that the order of the resamplings can be reversed (leading potentially to a locally different transcript).  The case $\scp(\phi_S)\cap\scp(\phi_{S'})=\Sbar\cap\overline{S'}=\emptyset$ is uninteresting.  So consider a variable $X_i$ with $i\in\Sbar\cap\overline{S'}$.  Necessarily, $i\in\partial S\cap\partial S'$.  It follows that $X_i$ takes the value~0 before $\scp(\phi_S)$ is resampled (since $\neg\phi_S$) and retains that value after (since $\neg\phi_{S'}$).  Therefore the two scopes could as well have been resampled in the opposite order.   In Figure~\ref{fig:clusters}, the two clusters are related in the dependency graph but not in the lopsided dependency graph: only the boundaries intersect.

Definition~\ref{def:lopsided} is sometimes portrayed as as a positive dependency condition, but in the resampling table view of the world it seems more natural to interpret it as a commutativity condition.  We say that a clause is \emph{atomic} if it is falsified by exactly one assignment.  Definition~\ref{def:lopsided} takes a simpler form when all clauses are atomic.
\begin{observation}
Suppose $\phi_k$ is atomic, for every $k\in[m]$.  Then $k\sim \ell$, i.e., $k$ and $\ell$ are lopsidedly dependent, iff $\phi_k\vee\phi_\ell$ is a tautology.
\end{observation}
Moser and Tardos~\cite[Thm 6.1]{MT10} prove the following runtime bound.

\begin{theorem}\label{thm:MTbd}
Suppose that $\Phi$ is an satisfiable quasi-extremal instance with lopsided dependency graph $([m],\sim)$.  Suppose also that there exists a sequence of reals $(x_k\in(0,1):k\in[m])$ such that, for all $k\in[m]$,
$$
\Pr_\calD(\neg\phi_k)\leq x_k\prod_{\ell\in[m]:\ell\sim k}(1-x_\ell).
$$
Then, in expectation, Algorthm~\ref{alg:PRS-nonextremal} resamples $\scp(\phi_k)$ at most $x_k/(1-x_k)$ times before halting.
\end{theorem}

\begin{lemma}
Suppose $G$ is a graph with $n$ vertices and maximum degree $\Delta$.  There exists $\lambda_\Delta>0$ such that the expected number of variable resamplings made during the execution of $\PRSforIS(G,\lambda)$ is $O(n)$ whenever $\lambda\leq\lambda_\Delta$.  Asymptotically, $\lambda_\Delta=\Theta(\Delta^{-1})$.
\end{lemma}

\begin{proof}
 Identifying vertices of $G$ with $[n]$, let 
$$
\calC=\big\{S\subseteq [n] :|S|\geq 2\text{ and }G[S]\text{ is connected}\big\}
$$
be the set of all subsets of $V(G)$ that induce connected subgraphs of $G$ with at least two vertices.  We need to find quantities $\{x_S:S\in\calC\}$ satisfying 
\begin{equation}\label{eq:XSbd}
\Pr_{\calD,\lambda}(\neg\phi_S)\leq x_S\prod_{S'\in\calC:S'\sim S}(1-x_{S'}).
\end{equation}
Note that we have included the activity $\lambda$ explicitly in the notation here, as we are about to introduce a second artificial activity~$\mu$.  We define the required quantities~$x_S$ by $x_S=\Pr_{\calD,\mu}(\neg\phi_S)$ for some suitably chosen $\mu$ (one that will make the right hand side of inequality~(\ref{eq:XSbd}) large), and then choose $\lambda<\mu$ as large as possible while still satisfying the inequality. The thinking here is that as $S$ varies, $x_S$ should shadow $\Pr_{\calD,\lambda}(\neg\phi_S)$, but with enough slack to allow inequality (\ref{eq:XSbd}) to be satisfied.  For convenience, let $q=\mu/(1+\mu)$.  We start with a preliminary calculation. For any $i\in[n]$,
\begin{align}
\sum_{S'\in\calC:i\in S'}x_{S'}&=\sum_{S'\in\calC:i\in S'}\Pr_{\calD,\mu}(\neg\phi_{S'})\notag\\
&=\Pr_{\calD,\mu}\left(\bigvee_{S'\in\calC:i\in S'}\neg\phi_{S'}\right)\label{eq:disjoint}\\
&=\Pr_{\calD,\mu}(\text{$i$ is contained in some cluster})\notag\\
&\leq \Delta q^2.\label{eq:incluster}
\end{align}
Equality~\eqref{eq:disjoint} follows from disjointness of the events $\neg\phi_{S'}$ over all $S'\in\calC$ with $i\in S'$ (by lopsided dependency).  Inequality \eqref{eq:incluster} is a simple upper bound on the probability that vertex~$i$ finds itself in a cluster.  

Now suppose that $S\in\calC$ and let $c=|S|$ and $b=|\partial S|$.  Then
\begin{align}
x_S\prod_{S'\sim S}(1-x_{S'})&\geq x_S\prod_{i\in S\cup\partial S}\,\prod_{S'\in\calC:i\in S'}(1-x_{S'})&\text{(by over-counting)}\label{eq:overcount}\\
&\geq x_S\prod_{i\in S\cup\partial S}\left(1-\sum_{S'\in\calC:i\in S'}x_{S'}\right)\notag\\
&\geq q^c(1-q)^b(1-\Delta q^2)^{b+c}&\text{(by \eqref{eq:incluster})}.\label{eq:rhs}
\end{align}  
This deals with the right hand side of (\ref{eq:XSbd}).  The left hand side is simply
\begin{equation}\label{eq:lhs}
\Pr_{\calD,\lambda}(\neg\phi_S)=p^c(1-p)^b,
\end{equation}
where $p$ stands for $\lambda/(1+\lambda)$.  Recall that we want to ensure that (\ref{eq:lhs}) is less than or equal to (\ref{eq:rhs}).  Since $q>p$, this goal is hardest to achieve, for any given~$c$, when $b$ is as large as possible.  Certainly $b\leq (\Delta-1)c$, so we assume $b=(\Delta-1)c$ from now on.  With this simplification, the inequality we wish to satisfy is 
$$
p^c(1-p)^{(\Delta-1)c}\leq q^c(1-q)^{(\Delta-1)c}(1-\Delta q^2)^{\Delta c},
$$
or, equivalently, 
\begin{equation}\label{eq:pandq}
p(1-p)^{(\Delta-1)}\leq q(1-q)^{(\Delta-1)}(1-\Delta q^2)^{\Delta}.
\end{equation}
We are free to choose $q$ as we like.  If we let $q=\Delta^{-1}$ then the right hand side is greater than $1/8\Delta$, enabling us to take $p=1/8\Delta$ and $\lambda=1/(8\Delta-1)$.  

Having shown that the premise of Theorem~\ref{thm:MTbd} holds, we can read off an upper bound on the expected number of resamplings.  Repeating an earlier trick,
\begin{align*}
\sum_{S\in\calC}x_S|S|&=\sum_{i\in V(G)}\sum_{S\in\calC:i\in S}x_{S}\\
&\leq n\Delta q^2&\text{(by \eqref{eq:incluster})}\\
&=n/\Delta.
\end{align*}
Noting that $|\partial S|\leq (\Delta-1)|S|$ and $x_S\leq\frac12$ we have that the expected total number of variable resamplings is
$$
\sum_{S\in\calC}\frac{x_S}{1-x_S}|S\cup\partial S|\leq \sum_{S\in\calC}2\Delta x_S|S|\leq 2n,
$$
by Theorem~\ref{thm:MTbd}.
\end{proof}
\begin{remark}
Of course, $\Delta^{-1}$ was merely a convenient choice for $q$ and not an optimal one.  When $\Delta=3$, we find numerically that the right hand side of (\ref{eq:pandq}) attains a maximum of $0.0892275+$ at around $q=0.172016$.  Thus, we can satisfy inequality (\ref{eq:pandq}) by setting $p=0.113551$, which its ensured by taking $\lambda_3=0.128$.  When $\Delta$ is large, a similar line of argument gives $\lambda_\Delta\sim C/\Delta$ asymptotically, where $x=C= 0.327+$ is the smallest solution to $xe^{-x}=\frac12e^{-3/4}$.  

Our calculation has some slack at a number of locations.  One easy win is to replace $\Delta q^2$ in (\ref{eq:incluster}) by the tighter, in fact exact, $q(1-(1-q)^\Delta)$.  Another arises from the following observation. Suppose $S\in\calC$ and $U_S\subset S$ is a minimum vertex cover in $G[S]$.  Then every $S'\in \calC$ with $S'\sim S$ either intersects $\partial S$ or $U_S$.  Thus, in (\ref{eq:overcount}) we may replace the sum over $i\in S\cup\partial S$ by a sum over $i\in U_S\cup\partial S$, and then replace $(1-\Delta q^2)^{b+c}$ in (\ref{eq:rhs}) by $(1-\Delta q^2)^{b/2+c}$, since a minimum vertex cover in $G[S]$ has size at most $\frac12|S|$.  Repeating the earlier calculation with these improvements we find that the right hand side of (\ref{eq:pandq}) achieves a maximum of $0.0990257+$ leading to $p=0.131189$ and $\lambda_3=0.150$.

There are several other steps of the calculation where slack is present, and could be reduced at the expense of additional combinatorial complexity.  
\end{remark}

To provide some context for the above working, we review the hard-core model on an infinite regular tree of degree~$\Delta$. It is known that this model exhibits a phase transition at $\lambda_c=(\Delta-1)^{\Delta-1}/(\Delta-2)^\Delta$. For $\lambda<\lambda_c$ there is a unique Gibbs measure and for $\lambda>\lambda_c$ there are two.  A remarkable discovery is that that $\lambda_c$ also marks a computational threshold of the hard-core model. On the one hand, Sly and Sun~\cite{SlySun} and Galanis, \v Stefankovi\v c and Vigoda~\cite{GSV16} showed that it is NP-hard to sample, even approximately, from the hard-core distribution in general graphs of maximum degree~$\Delta$, when $\lambda>\lambda_c$.  One the other hand, for the same class of graphs, approximate sampling is possible in time $O(n\log n)$ when $\lambda<\lambda_c$. This was shown by Chen, Liu and Vigoda~\cite{CLVoptimalIS}, building on the spectral independence approach of Anari, Liu and Oveis-Gharan~\cite{ALO-G}. 

Against this benchmark, the performance of PRS is unimpressive.  For $\Delta=3$, our $\lambda_3=0.150$ is woefully short of the computational threshold at $\lambda_c(3)=4$.  On the other hand, PRS is a perfect sampler and is certainly fast, making less than two resamplings per site in expectation.   The highest $\lambda$ for which linear-time perfect sampling is possible is unknown, but it is at least $\lambda=1/(\Delta-1)-\eps$, giving $\lambda=\frac12-\eps$ for $\Delta=3$~\cite[\S3.1]{AnandJerrum}.  It would be interesting to know whether the limit for linear time perfect sampling is $\lambda_c$, or whether there is a barrier below this.  

Extensive research on the algorithmic Lov\'asz Local Lemma has brought to light a number of alternatives to Theorem~\ref{thm:PRScorrect-nonextremal}. Examples which may be useful in analysing PRS algorithms have been given by  
Bissacot, Fern\'{a}ndez, Procacci, and Scoppola~\cite{BissacotEtAl},
Harris~\cite{Harris},
Harvey and Vondrak~\cite{HV15},
Iliopoulos~\cite{Iliopoulos},
Kolmogorov~\cite{Kolmogorov},
and Pegden~\cite{Pegden}. 
A comprehensive treatment of the circle of ideas surrounding the independent set polynomial and the Lov\'asz Local Lemma has been given by Scott and Sokal~\cite{ScottSokal}. 

\section{Generalisations}

In this article, we have restricted attention to the simplest version of PRS based directly on the Moser Tardos algorithmic LLL\null.  Specifically, we resample the variables of just one clause at each step.  This involved recasting the `obvious' encoding of a problem as a CNF formula in a form suitable for application of the method.  For example, in the case of independent sets, we replaced the natural two-variable clauses by larger clauses based on clusters.  Alternatively, it is possible to stick with the `natural' formula at the expense of complicating the resampling algorithm.  This was the approach originally taken by Guo, Jerrum and Liu~\cite{GJL19}.  

We dealt here exclusively with hard constraints which either permit or deny a particular assignment to the variables.  Soft constraints can be incorporated by introducing an auxiliary variable taking values in the real interval $[0,1]$.  Applied to the Ising model, for example, one would end up with a representation akin to that of Edwards and Sokal~\cite{EdwardsSokal}.   Alternatively, Feng, Vishnoi and Yin~\cite{FengVishnoiYin}  incorporated soft constraints directly, thereby allowing a wider range of spin systems to be addressed more naturally.  Another possible extension is to continuous state spaces, with Guo and Jerrum~\cite{GJ18b} treating the hard-disks model, and Moka and Kroese~\cite{MokaKroese} more general point processes.  Feng, Guo and Yin~\cite{FengGuoYin} show how to achieve perfect sampling when strong spatial mixing holds.  This last work is quite far from PRS, but still relies on growing a sample by repeatedly `repairing' parts of the current configuration.  

\section*{Acknowledgements}
The treatment of PRS presented here draws on many sources, in some cases heavily.  Particularly influential are the works of Moser and Tardos~\cite{MT10}, Knuth~\cite{KnuthVol4Fasc6}, Kolipaka and Szegedy~\cite{KS11} and Viennot~\cite{Viennot}.  I also learned a great deal through collaboration with Heng Guo.  Finally, in retrospect, it is remarkable how many of the ideas behind PRS were already present in the work of Propp and Wilson~\cite{PW98} on cycle-popping.

\bibliographystyle{plain}
\bibliography{notesOnRejectionSampling}

\begin{thebibliography}{10}

\bibitem{AnandJerrum}
Konrad Anand and Mark Jerrum.
\newblock Perfect sampling in infinite spin systems via strong spatial mixing.
\newblock {\em SIAM J. Comput.}, 51(4):1280--1295, 2022.

\bibitem{ALOV}
Nima Anari, Kuikui Liu, Shayan~Oveis Gharan, and Cynthia Vinzant.
\newblock Log-concave polynomials {II}: {H}igh-dimensional walks and an {FPRAS}
  for counting bases of a matroid.
\newblock In {\em S{TOC}'19---{P}roceedings of the 51st {A}nnual {ACM} {SIGACT}
  {S}ymposium on {T}heory of {C}omputing}, pages 1--12. ACM, New York, 2019.

\bibitem{ALO-G}
Nima Anari, Kuikui Liu, and Shayan Oveis~Gharan.
\newblock Spectral independence in high-dimensional expanders and applications
  to the hardcore model.
\newblock In {\em 2020 {IEEE} 61st {A}nnual {S}ymposium on {F}oundations of
  {C}omputer {S}cience}, pages 1319--1330. IEEE Computer Soc., Los Alamitos,
  CA, [2020] \copyright 2020.

\bibitem{BissacotEtAl}
Rodrigo Bissacot, Roberto Fern\'{a}ndez, Aldo Procacci, and Benedetto Scoppola.
\newblock An improvement of the {L}ov\'{a}sz local lemma via cluster expansion.
\newblock {\em Combin. Probab. Comput.}, 20(5):709--719, 2011.

\bibitem{CartierFoata}
P.~Cartier and D.~Foata.
\newblock {\em Probl\`emes combinatoires de commutation et r\'{e}arrangements}.
\newblock Lecture Notes in Mathematics, No. 85. Springer-Verlag, Berlin-New
  York, 1969.

\bibitem{CLVoptimalIS}
Zongchen Chen, Kuikui Liu, and Eric Vigoda.
\newblock Optimal mixing of {G}lauber dynamics: entropy factorization via
  high-dimensional expansion.
\newblock In {\em S{TOC} '21---{P}roceedings of the 53rd {A}nnual {ACM}
  {SIGACT} {S}ymposium on {T}heory of {C}omputing}, pages 1537--1550. ACM, New
  York, [2021] \copyright 2021.

\bibitem{CPP02}
Henry Cohn, Robin Pemantle, and James~G. Propp.
\newblock Generating a random sink-free orientation in quadratic time.
\newblock {\em Electr. J. Comb.}, 9(1), 2002.

\bibitem{EdwardsSokal}
Robert~G. Edwards and Alan~D. Sokal.
\newblock Generalization of the {F}ortuin-{K}asteleyn-{S}wendsen-{W}ang
  representation and {M}onte {C}arlo algorithm.
\newblock {\em Phys. Rev. D (3)}, 38(6):2009--2012, 1988.

\bibitem{Eriksson96}
Kimmo Eriksson.
\newblock Strong convergence and a game of numbers.
\newblock {\em European J. Combin.}, 17(4):379--390, 1996.

\bibitem{FengGuoYin}
Weiming Feng, Heng Guo, and Yitong Yin.
\newblock Perfect sampling from spatial mixing.
\newblock {\em Random Structures Algorithms}, 61(4):678--709, 2022.

\bibitem{FengVishnoiYin}
Weiming Feng, Nisheeth~K. Vishnoi, and Yitong Yin.
\newblock Dynamic sampling from graphical models.
\newblock {\em SIAM J. Comput.}, 50(2):350--381, 2021.

\bibitem{FillHuber}
James~Allen Fill and Mark Huber.
\newblock The randomness recycler: a new technique for perfect sampling.
\newblock In {\em 41st {A}nnual {S}ymposium on {F}oundations of {C}omputer
  {S}cience ({R}edondo {B}each, {CA}, 2000)}, pages 503--511. IEEE Comput. Soc.
  Press, Los Alamitos, CA, 2000.

\bibitem{GSV16}
Andreas Galanis, Daniel \v{S}tefankovi\v{c}, and Eric Vigoda.
\newblock Inapproximability of the partition function for the antiferromagnetic
  {I}sing and hard-core models.
\newblock {\em Combin. Probab. Comput.}, 25(4):500--559, 2016.

\bibitem{GP14}
Igor Gorodezky and Igor Pak.
\newblock Generalized loop-erased random walks and approximate reachability.
\newblock {\em Random Structures Algorithms}, 44(2):201--223, 2014.

\bibitem{GH18}
Heng Guo and Kun He.
\newblock Tight bounds for popping algorithms.
\newblock {\em Random Structures Algorithms}, 57(2):371--392, 2020.

\bibitem{GJ19}
Heng Guo and Mark Jerrum.
\newblock A polynomial-time approximation algorithm for all-terminal network
  reliability.
\newblock {\em SIAM J. Comput.}, 48(3):964--978, 2019.

\bibitem{GJ18c}
Heng Guo and Mark Jerrum.
\newblock Approximately counting bases of bicircular matroids.
\newblock {\em Combin. Probab. Comput.}, 30(1):124--135, 2021.

\bibitem{GJ18b}
Heng Guo and Mark Jerrum.
\newblock Perfect simulation of the hard disks model by partial rejection
  sampling.
\newblock {\em Ann. Inst. Henri Poincar\'{e} D}, 8(2):159--177, 2021.

\bibitem{GJL19}
Heng Guo, Mark Jerrum, and Jingcheng Liu.
\newblock Uniform sampling through the {L}ov\'{a}sz local lemma.
\newblock {\em J. ACM}, 66(3):Art. 18, 31, 2019.

\bibitem{Harris}
David~G. Harris.
\newblock Lopsidependency in the {M}oser-{T}ardos framework: beyond the
  lopsided {L}ov\'{a}sz local lemma.
\newblock {\em ACM Trans. Algorithms}, 13(1):Art. 17, 26, 2016.

\bibitem{HV15}
Nicholas J.~A. Harvey and Jan Vondr\'{a}k.
\newblock An algorithmic proof of the {L}ov\'{a}sz local lemma via resampling
  oracles.
\newblock {\em SIAM J. Comput.}, 49(2):394--428, 2020.

\bibitem{Iliopoulos}
Fotis Iliopoulos.
\newblock Commutative algorithms approximate the {LLL}-distribution.
\newblock In {\em Approximation, randomization, and combinatorial optimization.
  {A}lgorithms and techniques}, volume 116 of {\em LIPIcs. Leibniz Int. Proc.
  Inform.}, pages Art. No. 44, 20. Schloss Dagstuhl. Leibniz-Zent. Inform.,
  Wadern, 2018.

\bibitem{JVW}
F.~Jaeger, D.~L. Vertigan, and D.~J.~A. Welsh.
\newblock On the computational complexity of the {J}ones and {T}utte
  polynomials.
\newblock {\em Math. Proc. Cambridge Philos. Soc.}, 108(1):35--53, 1990.

\bibitem{KK17}
Adrien Kassel and Richard Kenyon.
\newblock Random curves on surfaces induced from the {L}aplacian determinant.
\newblock {\em Ann. Probab.}, 45(2):932--964, 2017.

\bibitem{KnuthVol4Fasc6}
Donald~E. Knuth.
\newblock {\em The Art of Computer Programming}, volume 4, Fascicle 6.
\newblock Addison-Wesley Professional, 2015.

\bibitem{KS11}
Kashyap Kolipaka and Mario Szegedy.
\newblock Moser and {T}ardos meet {L}ov\'{a}sz.
\newblock In {\em S{TOC}'11---{P}roceedings of the 43rd {ACM} {S}ymposium on
  {T}heory of {C}omputing}, pages 235--243. ACM, New York, 2011.

\bibitem{Kolmogorov}
Vladimir Kolmogorov.
\newblock Commutativity in the algorithmic {L}ov\'{a}sz local lemma.
\newblock {\em SIAM J. Comput.}, 47(6):2029--2056, 2018.

\bibitem{LubyVigoda}
Michael Luby and Eric Vigoda.
\newblock Fast convergence of the {G}lauber dynamics for sampling independent
  sets.
\newblock {\em Random Structures Algorithms}, 15(3-4):229--241, 1999.

\bibitem{MokaKroese}
Sarat~B. Moka and Dirk~P. Kroese.
\newblock Perfect sampling for {G}ibbs point processes using partial rejection
  sampling.
\newblock {\em Bernoulli}, 26(3):2082--2104, 2020.

\bibitem{MT10}
Robin~A. Moser and G{\'{a}}bor Tardos.
\newblock A constructive proof of the general {L}ov{\'{a}}sz {L}ocal {L}emma.
\newblock {\em J. {ACM}}, 57(2), 2010.

\bibitem{Pegden}
Wesley Pegden.
\newblock An extension of the {M}oser-{T}ardos algorithmic local lemma.
\newblock {\em SIAM J. Discrete Math.}, 28(2):911--917, 2014.

\bibitem{PW98}
James~G. Propp and David~B. Wilson.
\newblock How to get a perfectly random sample from a generic {M}arkov chain
  and generate a random spanning tree of a directed graph.
\newblock {\em J. Algorithms}, 27(2):170--217, 1998.

\bibitem{ScottSokal}
Alexander~D. Scott and Alan~D. Sokal.
\newblock The repulsive lattice gas, the independent-set polynomial, and the
  {L}ov\'{a}sz local lemma.
\newblock {\em J. Stat. Phys.}, 118(5-6):1151--1261, 2005.

\bibitem{Shearer85}
James~B. Shearer.
\newblock On a problem of {S}pencer.
\newblock {\em Combinatorica}, 5(3):241--245, 1985.

\bibitem{SlySun}
Allan Sly and Nike Sun.
\newblock Counting in two-spin models on {$d$}-regular graphs.
\newblock {\em Ann. Probab.}, 42(6):2383--2416, 2014.

\bibitem{Viennot}
G\'{e}rard~Xavier Viennot.
\newblock Heaps of pieces. {I}. {B}asic definitions and combinatorial lemmas.
\newblock In {\em Combinatoire \'{e}num\'{e}rative ({M}ontreal, {Q}ue.,
  1985/{Q}uebec, {Q}ue., 1985)}, volume 1234 of {\em Lecture Notes in Math.},
  pages 321--350. Springer, Berlin, 1986.

\end{thebibliography}

\end{document}